\documentclass[10pt, twosided]{article}

\usepackage{amssymb, amsfonts, amsmath, amsthm, mathrsfs, verbatim,graphicx,graphics,epsfig,psfrag,mathtools, bm, bbm}

\usepackage{authblk}

\DeclareMathAlphabet{\mathpzc}{OT1}{pzc}{m}{it}

\usepackage{mathabx}
\usepackage{microtype} 

\usepackage{color, xcolor}

\usepackage{accents}
\newcommand{\ubar}[1]{{\underaccent{\bar}{\smash{#1}}}}

\usepackage{abstract}
\usepackage[all]{xy}

\usepackage[inline]{enumitem}

\usepackage{hyperref}
\makeatletter 
\@addtoreset{equation}{section}
\makeatother  

\numberwithin{equation}{section}

\usepackage{tikz-cd}
\usepackage{tikz}
\usetikzlibrary{arrows}
\usepackage{float}

\usetikzlibrary{decorations.pathreplacing}
\usepackage{etex}

\makeatletter
\def\@tocline#1#2#3#4#5#6#7{\relax
  \ifnum #1>\c@tocdepth 
  \else
    \par \addpenalty\@secpenalty\addvspace{#2}%
    \begingroup \hyphenpenalty\@M
    \@ifempty{#4}{%
      \@tempdima\csname r@tocindent\number#1\endcsname\relax
    }{%
      \@tempdima#4\relax
    }%
    \parindent\z@ \leftskip#3\relax \advance\leftskip\@tempdima\relax
    \rightskip\@pnumwidth plus4em \parfillskip-\@pnumwidth
    #5\leavevmode\hskip-\@tempdima
      \ifcase #1
       \or\or \hskip 1em \or \hskip 2em \else \hskip 3em \fi%
      #6\nobreak\relax
    \dotfill\hbox to\@pnumwidth{\@tocpagenum{#7}}\par
    \nobreak
    \endgroup
  \fi}
\makeatother

\newtheorem{thm}{Theorem}[section]
\newtheorem{cor}[thm]{Corollary}

\newtheorem{conjecture}[thm]{Conjecture}
\newtheorem{prop}[thm]{Proposition}
\newtheorem{lemma}[thm]{Lemma}

\theoremstyle{definition}
\newtheorem{defn}[thm]{Definition}
\theoremstyle{remark}
\newtheorem{rem}[thm]{Remark}
\newtheorem{hyp}[thm]{Hypothesis}

\newtheorem{notation}[thm]{Notation}

\newcounter{notes}
{\end{list}}

\makeatletter
\@addtoreset{thm}{section}
\makeatother

\newcommand\qu{/\kern-.7ex/} 
\renewcommand{\setminus}{\smallsetminus}

\newcommand{\beq}{\begin{equation}}
\newcommand{\eeq}{\end{equation}}
\newcommand{\beqn}{\begin{equation*}}
\newcommand{\eeqn}{\end{equation*}}
\newcommand{\ov}{\overline}
\newcommand{\mb}{\mathbb}
\renewcommand{\tt}{\texttt}
\newcommand{\mc}{\mathcal}
\newcommand{\mf}{\mathfrak}

\newcommand{\V}{{\rm V}}
\newcommand{\T}{{\rm T}}
\newcommand{\E}{{\rm E}}

\newcommand{\preq}{\preccurlyeq}

\renewcommand{\i}{{\bf i}}

\newcommand{\Gammait}{{\mathit{\Gamma}}}

\newcommand{\Thetait}{{\mathit{\Theta}}}

\title{Gauged Linear Sigma Model in the Geometric Phase. I.}

\begin{document}

\author[Tian]{Gang Tian}

\affil[Tian]{BICMR and SMS, Beijing University, 100871, P.R.China}

\author[Xu]{Guangbo Xu
\thanks{The second named author is partially supported by the Simons Foundation through the Homological Mirror Symmetry Collaboration grant and NSF DMS-2345030.}
}
\affil[Xu]{Department of Mathematics, Rutgers University, Piscataway, NJ 08854-8019, USA}



\setcounter{tocdepth}{2}


\maketitle

\begin{abstract}
We construct correlation functions a gauged linear sigma model space in the geometric phase and prove that they satisfy the splitting axioms of the cohomological field theory. The method we use is from symplectic geometry and gauge theory, rather than algebraic geometry. These correlation functions are expected to match the Gromov--Witten invariants of the classical vacuum up to a change of variable, and is expected to match various other algebraic geometric constructions.
\end{abstract}



\tableofcontents

\section{Introduction}

The gauged linear sigma model (GLSM) is a two-dimensional supersymmetric quantum field theory introduced by Witten \cite{Witten_LGCY}. It has stimulated many important developments in both the mathematical and the physical studies of string theory and mirror symmetry. For example, it plays a fundamental role in physicists' argument of mirror symmetry \cite{Hori_Vafa}\cite{Gu_Sharpe_2018}. Its idea is also of crucial importance in the verification of genus zero mirror symmetry for quintic threefold \cite{Givental_96}\cite{LLY_1}. \footnote{By the name of ``gauged linear sigma model,'' we always mean the case that the superpotential $W$ is nonzero. On the technical level, a nonzero $W$ makes a significant difference from the case when $W = 0$. However, in many literature, particularly physics literature, GLSM also include the case that $W = 0$ and many works only deal with this simpler case.} 

Since 2012 the authors have initiated a project aiming at constructing a mathematically rigorous theory for GLSM, using mainly the method from symplectic and differential geometry. In our previous works \cite{Tian_Xu, Tian_Xu_2, Tian_Xu_3}, we have constructed certain correlation function (i.e. Gromov--Witten type invariants) under certain special conditions: the gauge group is $U(1)$ and the superpotential is a Lagrange multiplier. This correlation function is though rather restricted, as for example, we do not know if they satisfy splitting properties or not. 

A major difficulty in the study of the gauged Witten equation is in the so-called {\it broad} case, in which the Fredholm property of the equation is problematic. This issue also causes difficulties in various algebraic approaches. We realized that if we restrict to the so-called geometric phase, then the issue about broad case disappears in the symplectic setting. This idea has been outlined in \cite{Tian_Xu_2017} and the detailed construction is given in the current and the companion paper \cite{Tian_Xu_geometric_2}.

A natural question is the relation between the GLSM correlation functions and the Gromov--Witten invariants. They are both multilinear functions on the cohomology group of the classical vacuum of the GLSM space, which is a compact K\"ahler manifold. The relation between them is a generalization of the {\it quantum Kirwan map}, proposed by D. Salamon, 
proved by Woodward \cite{Woodward_15} in the algebraic case (also see the recent work by the second-named author \cite{Xu_Steenrod} in the symplectic setting under the monotonicity assumption). In short words, the two types of invariants are related by the ``mirror map.'' The mirror map is defined by counting pointlike instantons, which are solutions to the gauged Witten equation over the complex plane ${\mb C}$. This picture also has a physics explanation by Morrison--Plesser \cite{Morrison_Plesser_1995}. The enumerative definition of the mirror map and the proof of the relation between GLSM correlation functions and Gromov--Witten invariants will be provided in a forthcoming paper \cite{Limit2}.

\subsection{The main theorem}\label{subsection11}

In this paper we define Gromov--Witten type invariants of the geometric objects called {\it gauged linear sigma model} (GLSM) spaces. A {\it GLSM space} is a quadruple $(V, G, W, \mu)$ where $V$ is a K\"ahler manifold (usually noncompact), $G$ is a reductive Lie group with maximal compact subgroup $K$ acting on $V$, $W$ (the {\it superpotential}) is a $G$-invariant holomorphic function on $V$, and $\mu$ is a moment map of the $K$-action, which specifies a stability condition. An important hypothesis of this paper is the {\it geometric phase} condition which implies that the classical vacuum
\beqn
X:= {\rm Crit} W\qu G \cong ({\rm Crit} W \cap \mu^{-1}(0))/K
\eeqn
is a compact K\"ahler manifold. In a typical example (see Subsection \ref{subsection32}), $X$ is a hypersurface in the projective space, including the famous quintic three-fold.

The invariants (which we call GLSM correlation functions) are defined via the virtual intersection theory over the moduli space of the {\it gauged Witten equation}. The domains of this equation are $r$-spin curves (where $r$ is the degree of the superpotential). An {\it $r$-spin curve} is a quadruple ${\mc C} = (\Sigma, {\bf z}, L, \varphi)$ where $\Sigma$ is an orbifold curve, ${\bf z}$ is the set of orbifold marked points, $L$ is an orbifold line bundle over $\Sigma$, and $\varphi$ is an isomorphism from $L^{\otimes r}$ to the log-canonical bundle of $(\Sigma, {\bf z})$ (see details in Section \ref{section2}). The variables of the gauged Witten equation are called gauged maps. A {\it gauged map} from an $r$-spin curve ${\mc C}$ to the GLSM space $(V, G, W, \mu)$ is a triple $(P, A, u)$, where $P$ is a principal $K$-bundle over $\Sigma$, $A$ is a connection on $P$, and $u$ is a section of the fibre bundle over $\Sigma$ associated to $P$ and $L$ (with structure group being $K \times U(1)$ and fibre being $V$). A gauged map induces the following objects: 
\begin{enumerate}
\item The covariant derivative $d_A u$ and its $(0,1)$-part $\ov\partial_A u$.

\item The curvature $F_A$.

\item The moment map potential $\mu(u)$.

\item The gradient $\nabla W(u)$.
\end{enumerate}
The gauged Witten equation reads
\begin{align*}
&\  \ov\partial_A u  + \nabla W (u) = 0,\ &\ * F_A + \mu(u) = 0.
\end{align*}
The gauged Witten equation is a generalization of the {\it symplectic vortex equation} introduced by Mundet \cite{Mundet_thesis, Mundet_2003} and Cieliebak--Gaio--Salamon \cite{Cieliebak_Gaio_Salamon_2000}. 

A crucial point in our setting is that we use {\it cylindrical} metrics on domain curves. As a consequence, solutions have well-defined evaluations at marked points in the classical vacuum. When $(V, G, W, \mu)$ is in {\it geometric phase} (Definition \ref{defn33}) the gauged Witten equation has a Morse--Bott type asymptotic behavior near marked points. This feature resolves the previous difficulty in the broad case of the gauged Witten equation studied in \cite{Tian_Xu, Tian_Xu_2, Tian_Xu_3}.

The remaining construction can be carried out via the standard, if not at all simple, procedure. One can define a natural compactification $\ov{\mc M}{}_{g, n}^r(V, G, W, \mu)$ of the moduli space of gauge equivalence classes of stable solutions. There is a continuous evaluation map 
\beqn
\ov{\mc M}{}_{g, n}^r(V, G, W, \mu) \to X^n \times \ov{\mc M}_{g, n}.
\eeqn
The Fredholm property of the equation allows us to carry out the sophisticated virtual cycle construction (which we give more detailed remarks in Subsection \ref{subsection12}). The whole construction can be concluded in the following main theorem. Let $\Lambda$ be the Novikov field
\beqn
\Lambda:= \Big\{ \sum_{i=1}^\infty a_i T^{\lambda_i}\ |\ a_i \in {\mb Q},\ \lambda_i \in {\mb R},\ \lim_{i \to\infty} \lambda_i = +\infty \Big\}.
\eeqn

\begin{thm}\label{thm11}
For a GLSM space $(V, G, W, \mu)$ in the geometric phase, there are a collection of multilinear functions
\beqn
\langle \cdot \rangle_{g,n}^{\rm GLSM}: \underbrace{H^*(X; \Lambda) \otimes \cdots \otimes H^*(X; \Lambda)}_{n} \otimes H^*( \ov{\mc M}_{g,n}; \Lambda) \to \Lambda,\ 2g + n \geq 3.
\eeqn
defined by the virtual integration over the moduli $\ov{\mc M}{}_{g,n}^r(V, G, W, \mu)$. Moreover, the correlation functions satisfying the following properties:

\begin{enumerate}

\item {\bf (Splitting Property for Non-Separating Nodes)} Suppose $2g + n > 3$. Let $\gamma \in H^2(\ov{\mc M}_{g,n}; {\mb Q})$ be the class dual to the divisor of configurations obtained by shrinking a non-separating loop, which is the image of a map $\iota_\gamma: \ov{\mc M}_{g-1, n+2} \to \ov{\mc M}_{g, n}$. Then 
\beqn
\langle \alpha_1 \otimes \cdots \otimes \alpha_n; \beta \cup \gamma \rangle_{g, n}^{\rm GLSM} = \langle \alpha_1 \otimes \cdots \otimes \alpha_n \otimes \Delta; \iota_\gamma^* \beta \rangle_{g-1, n+2}^{\rm GLSM}.
\eeqn
Here $\Delta = \sum_j \delta_j \otimes \delta^j \in H^*(X \times X)$ is the Poincar\'e dual to the diagonal of $X$.

\item {\bf (Splitting Property for Separating Nodes)} Let $\gamma \in H^2(\ov{\mc M}_{g, n}; {\mb Q})$ be class dual to the divisor of configurations obtained by shrinking a separating loop, which is the image of a map $\iota_\gamma: \ov{\mc M}_{g_1, n_1 + 1} \times \ov{\mc M}_{g_2, n_2 + 1} \to \ov{\mc M}_{g, n}$, also characterized by a decomposition $\{1, \ldots, n\} = I_1 \sqcup I_2$ with $|I_1| = n_1$, $|I_2| = n_2$. Suppose $2g_i  - 2 + n_i \geq 0$. Then for $\beta \in H^*(\ov{\mc M}_{g, n}; {\mb Q})$, if we write $\iota_\gamma^* \beta = \sum_l \beta_{1, l} \otimes \beta_{2, l}$ by K\"unneth decomposition, then 
\begin{multline*}
\langle \alpha_1 \otimes \cdots \otimes \alpha_n; \beta \cup \gamma \rangle_{g, n}^{\rm GLSM} \\
= \epsilon(I_1, I_2) \sum_{j, l} \langle \alpha_{I_1} \otimes \delta_j; \beta_{1, l} \rangle_{g_1, n_1+1}^{\rm GLSM} \langle \alpha_{I_2} \otimes \delta^j; \beta_{2, l} \rangle_{g_2, n_2+1}^{\rm GLSM}.
\end{multline*}
Here $\epsilon(I_1, I_2)$ is the sign of permutations of odd-dimensional $\alpha_i$'s.
\end{enumerate}
\end{thm}

The splitting properties of the correlation functions verify the most important axioms of the Cohomological Field Theory (see \cite{Manin_2}). 

\begin{rem}
Besides our symplectic approaches, there are other approaches toward mathematical theories of the GLSM using algebraic geometry. These include the quasimap theory developed by Ciocan-Fontanine--Kim--Maulik \cite{CK_2010, CKM_quasimap, CK_quasimap, CCK_quasimap, CK_2017, CK_bigi, CK_2020}, the theory of Mixed-Spin-P fields by Chang--Li--Li--Liu \cite{CLLL_2019, CLLL_16} (for the Fermat quintic), the GLSM theory of Fan--Jarvis--Ruan \cite{FJR_GLSM}, as well as the categorical approach of Ciocan-Fontanine {\it et al.} \cite{CFGKS}.

The general advantage of our approach lies in at least two aspects. First, as emphasized above, our approach (as well as that of \cite{CFGKS}) allows one to have broad insertions, allowing us to obtain a cohomological theory. In contrast, the parallel algebraic geometric approach of \cite{FJR_GLSM} cannot deal with broad insertions. Second, our method can be generalized the open-string situation where we can impose Lagrangian boundary condition; it is also possible to extend our setting of the gauged Witten equation to non-K\"ahler case if appropriate conditions of the superpotential and the action are imposed.
\end{rem}

\subsection{Virtual cycle}\label{subsection12}

In the companion paper \cite{Tian_Xu_geometric_2} we provide the detailed virtual cycle construction (as well as a few other technical results) needed for the definition of the GLSM correlation functions. In symplectic geometry, the virtual cycle theory arose in mathematicians' efforts in defining Gromov--Witten invariants for general symplectic manifolds. Jun Li and the first named author have their approach in both the algebraic case \cite{Li_Tian_2} and the symplectic case \cite{Li_Tian}. Meanwhile there were also other approaches such as the method of Kuranishi structure of Fukaya--Ono \cite{Fukaya_Ono} (further expanded by Fukaya--Oh--Ohta--Ono \cite{FOOO_Kuranishi}). Recently there have been various new developments such as the polyfold method of Hofer--Wysocki--Zehnder \cite{HWZ1, HWZ2, HWZ3} and the algebraic topological approach of Pardon \cite{Pardon_virtual} and Abouzaid \cite{Abouzaid_axiomatic}. Our virtual cycle construction essentially follows the topological viewpoint of Li--Tian \cite{Li_Tian}.

\subsection{Outline}

In Section \ref{section2} we review the basic notion about $r$-spin curves and results about the existence of cylindrical metrics. In Section \ref{section3} we set up the gauged Witten equation. In Section \ref{section4} we prove a few important analytical properties of solutions to the gauged Witten equation. In Section \ref{section5} we give the compactification of the moduli spaces. In Section \ref{section6} we state the results about the virtual cycle and construct the GLSM cohomological field theory.

\subsection{Acknowledgement}

We thank Wei Gu, Mauricio Romo, and Jake Solomon for stimulating discussions. We thank Alexander Kupers for kindly answering questions about topological transversality. 

\section{Moduli Spaces of Stable $r$-Spin Curves}\label{section2}

In this section we describe the moduli space of stable $r$-spin curves from the differential-geometric point of view. One can compare with the algebraic approaches in for example
\cite{Jarvis_1998, Jarvis_2000}\cite{Jarvis_Kimura_Vaintrob}\cite{Abramovich_Jarvis_2002}\cite{Chiodo_2008}, or the more general case of $W$-curves in \cite{FJR_annals}.

\subsection{$r$-spin curves}\label{subsection21}

In this discussion there are three categories of curves, ordinary complex curves (Riemann surfaces), orbifold curves, and $r$-spin curves. To unify notations, orbifold curves are usually denoted by $\Sigma$ while $r$-spin curves are usually denoted by ${\mc C}$. When we need to discuss their underying coarse curves, we usually still use $\Sigma$.

\subsubsection{Orbifold curves and orbifold bundles}

A one-dimensional complex orbifold is called an orbifold Riemann surface or an orbifold curve. In this paper, we impose the following conditions and conventions.

\begin{enumerate}

\item Orbifold curves are effective orbifolds with finitely many orbifold points. 

\item We always assume that an orbifold curve is marked. Moreover, the set of orbifold points are contained in the set of markings. However, a marking may not be a strict orbifold point.
\end{enumerate}

Let $\Sigma$ be an orbifold curve with the ordered set of markings $\vec{\bf z} = (z_1, \ldots, z_n)$. Near each $z_a$ there exists an orbifold chart of the form 
\beqn
(U_a, \Gammait_a) \cong ({\mb D}, {\mb Z}_{r_a}),\ r_a \geq 1.
\eeqn
Here ${\mb D}\subset {\mb C}$ is the open unit disk and ${\mb Z}_{r_a} \subset U(1)$ acts on $\Sigma$ in the standard fashion.

The notion of nodal orbifold curves is slightly more complicated. At a node $w$ a nodal orbifold curve has a chart of the form
\beq\label{nodalchart}
(U_w, \Gammait_w) \cong \big( \{ (\xi_-, \xi_+) \in {\mb D}^2\ |\ \xi_- \xi_+ = 0 \big\}, {\mb Z}_{r_w} \big),\ r_w \geq 1,
\eeq
where the ${\mb Z}_{r_w}$-action on $U_w$ is given by 
\beqn
\gamma (\xi_-, \xi_+) = (\gamma^{-1} \xi_-, \gamma \xi_+).
\eeqn

Given a smooth orbifold curve $\Sigma$, an orbifold line bundle is $L$ is a complex orbifold with a holomorphic map $\pi: L \to \Sigma$ which, over non-orbifold points has local trivializations as an ordinary holomorphic line bundle, while over an orbifold chart $(U_a, \Gammait_a)$ there is a chart of the form 
\beqn
(\tilde U_a, \Gammait_a) \cong ( U_a \times {\mb C}, {\mb Z}_{r_a})
\eeqn
where the ${\mb Z}_{r_a}$-action on the ${\mb C}$-factor is given by a weight $n_a$, i.e. 
\beqn
\gamma( z, t) = (\gamma z, \gamma^{n_a} t).
\eeqn
We call the element $e^{2\pi q {\bf i}} \in {\mb Z}_r$, where $q = \frac{ n_a}{r_a}$ the {\it monodromy} of the line bundle at $z_a$. When $\Sigma$ is nodal, we require that an orbifold line bundle $L \to \Sigma$ has local charts at a node $w$ of the form 
\beq\label{eqn22}
(\tilde U_w, \Gammait_w) \cong \big( \{ (\xi_-, \xi_+, t) \in {\mb D} \times {\mb D} \times {\mb C}\ |\ \xi_- \xi_+ = 0\}, {\mb Z}_{r_w} \big)
\eeq
where the ${\mb C}$-factor still has the linear action by the local group ${\mb Z}_{r_w}$. Then with respect to the normalization map $\tilde \Sigma \to \Sigma$, the pull-back bundle is an orbifold line bundle whose monodromies at preimages of a node are opposite. 

To define $r$-spin curves, one needs the notion of log-canonical bundle. Let $\Sigma$ be a smooth or nodal Riemann surface and $\pi: \tilde \Sigma \to \Sigma$ be the normalization. Let $\tilde w_1, \ldots, \tilde w_m \in \tilde \Sigma$ be the preimages of nodal points. The canonical bundle $K_\Sigma \to \Sigma$ is a well-defined line bundle; there is an isomorphism
\beqn
\pi^* K_\Sigma \cong K_{\tilde \Sigma} \otimes \bigotimes_{b=1}^m {\mc O}(\tilde w_b).
\eeqn

\begin{defn}{\rm (Log-canonical bundle)}
Let $(\Sigma, \vec{\bf z})$ be a smooth or nodal Riemann surface with markings $z_1, \ldots, z_n$. Its {\it log-canonical bundle} is the line bundle
\beqn
K_{\Sigma, {\rm log}} = K_\Sigma \otimes \bigotimes_{a=1}^n {\mc O}( z_a).
\eeqn
For an orbifold curve $\Sigma$, its log-canonical bundle is the pull-back of the log-canonical bundle of its underlying coarse curve.
\end{defn}

\begin{rem}
When we discuss a single marked curve $(\Sigma, \vec{\bf z})$, over the complement of markings and nodes we can trivialize ${\mc O}(z_a)$ hence $K_{\Sigma, {\rm log}}$ can be regarded as the canonical bundle of the punctured surface. However, when we discuss a family of curves, there is no canonical family of trivializations of the bundles ${\mc O}(z_a)$ over the punctured surface. We make such choices in Subsection \ref{subsection23}. 
\end{rem}

\subsubsection{$r$-spin curves}

Now we introduce a central concept of this paper, called $r$-spin curves. 

\begin{defn}\label{defn23} {\rm ($r$-spin curve)} Let $r$ be a positive integer.
\begin{enumerate}
\item A smooth or nodal $r$-spin curve of type $(g, n)$ is a quadruple ${\mc C} = (\Sigma_{\mc C}, \vec{\bf z}_{\mc C}, L_{\mc C}, \varphi_{\mc C})$ where $(\Sigma_{\mc C}, \vec{\bf z}_{\mc C})$ is a genus $g$, $n$-marked smooth or nodal orbifold curve, $L_{\mc C} \to \Sigma_{\mc C}$ is a holomorphic orbifold line bundle, and $\varphi_{\mc C}$ is an orbibundle isomorphism 
\beqn
\varphi_{\mc C}: L_{\mc C}^{\otimes r} \cong K_{\Sigma_{\mc C}, {\rm log}}.
\eeqn
We usually denote an $r$-spin curve as a quadruple ${\mc C} = (\Sigma_{\mc C}, \vec{\bf z}_{\mc C}, L_{\mc C}, \varphi_{\mc C})$. The log-canonical bundle of $\Sigma_{\mc C}$ is often denoted by $K_{{\mc C}, {\rm log}}$.

\item Let ${\mc C}_i = (\Sigma_{{\mc C}_i}, \vec{\bf z}_{{\mc C}_i}, L_{{\mc C}_i}, \varphi_{{\mc C}_i})$, $i = 1, 2$ be two $r$-spin curves. An {\it isomorphism} from ${\mc C}_1$ to ${\mc C}_2$ consists of an isomorphism of orbifold bundles (represented by the following commutative diagram)
\beqn
\vcenter{ \xymatrix{ L_{ {\mc C}_1} \ar[r]^{\tilde\rho} \ar[d]_{\pi_1} & L_{{\mc C}_2} \ar[d]^{\pi_2} \\
             \Sigma_{{\mc C}_1}   \ar[r]^\rho    & \Sigma_{{\mc C}_2}     } }
\eeqn
such that the following induced diagram commutes
\beq\label{eqn23}
\vcenter{ \xymatrix{ L_{{\mc C}_1}^{\otimes r} \ar[r]^{(\tilde\rho)^{\otimes r}}  \ar[d]_{\varphi_{{\mc C}_1}} & L_{{\mc C}_2}^{\otimes r} \ar[d]^{\varphi_{{\mc C}_2}   } \\ 
             K_{{{\mc C}_1}, {\rm log}}\ar[r]^{\rho} & K_{{{\mc C}_2}, {\rm log}}} }.
\eeq

\end{enumerate}

\end{defn}

\begin{rem} {\rm (Minimality of the local group)}
We require that the orbifold structures are the markings or nodes of ${\mc C}$ are {\it minimal} in the following sense. Take a local chart of $L_{\mc C}$ at a marking $z_a$
\beqn
({\mb D} \times {\mb C}, {\mb Z}_{r_a})
\eeqn
where the action of ${\mb Z}_{r_a}$ on the ${\mb C}$-factor has weight $n_a$. Then the isomorphism $\varphi_{\mc C}: L_{\mc C}^{\otimes r} \cong K_{{\mc C}, {\rm log}}$ implies that $r n_a/ r_a$ is an integer, hence there exists an integer $m_a \in \{0, 1, \ldots, r-1\}$ such that
\beqn
\frac{ n_a}{r_a} = \frac{m_a}{r}.
\eeqn
By abuse of notations, the integer $m_a$ is also called the {\it monodromy} of the $r$-spin structure at $z_a$. We require that $n_a$ and $r_a$ are coprime. Similar requirement is imposed for nodes. In other words, the group generated by the monodromies of $L_{\mc C}$ at a marking or a node is the same as the local group of the orbifold curve at that marking or node. In particular, when the monodromy of $L_{\mc C}$ is trivial, the marking or node is not a strict orbifold point. 
\end{rem}

\subsubsection{Infinite cylinders}

The only unstable $r$-spin curves we consider in this paper are topologically spheres with two marked points. As the complement of markings is a cylinder we also refer them to as infinite cylinders. These rational curves are classified by an element $m \in \{0, 1, \ldots, r-1\}$. For $m$ define $n_m$ and $r_m$ to be smallest nonnegative integers such that 
\beqn
\frac{n_m}{r_m} = \frac{m}{r}.
\eeqn
There is an $r$-spin curve whose underlying orbifold curve is 
\beqn
(\mb{CP}^1/{\mb Z}_{r_m}, 0, \infty)
\eeqn
and the orbifold line bundle is 
\beqn
L_m:= (\mb{CP}^1 \times {\mb C})/ {\mb Z}_{r_m}
\eeqn
where ${\mb Z}_{r_m}$ acts on the ${\mb C}$-factor with weight $n_m$. Up to isomorphism of $r$-spin structures there is only one isomorphism $\varphi_m: L_m^{\otimes r} \cong K_{{\rm log}}$. 

It is often more convenient to view these rational curves from the cylindrical perspective. Notice that the complement of the marked points $0$ and $\infty$ is isomorphic to the smooth curve
\beqn
\Thetait:= {\mb R} \times S^1.
\eeqn
Let $w = s + \i t$ be the standard cylindrical coordinate. The log-canonical bundle restricted to $\Thetait$ is trivial with a canonical section $dw$. Then there is a holomorphic section $e: \Thetait \to L_m$ such that 
\beqn
\varphi_m( e^{\otimes r}) = e^{m w} dw. 
\eeqn

\subsection{Moduli space of stable $r$-spin curves}\label{subsection22}

Every $r$-spin curve has an underlying marked smooth or nodal curve defined by forgetting the $r$-spin structure and the orbifold structure. 

\begin{defn}[Stable $r$-spin curves]
An $r$-spin curve is called {\it stable} if the underlying marked smooth or nodal curve is stable.
\end{defn}

Indeed, an $r$-spin curve is stable if and only if it has a finite automorphism group. Let $\ov{\mc M}{}_{g,n}^r$ denote the moduli space of stable $n$-marked $r$-spin curves of genus $g$; for a stable $n$-marked $r$-spin curve ${\mc C}$ of genus $g$, denote by $[{\mc C}] \in \ov{\mc M}{}_{g,n}^r$ the point represented by ${\mc C}$. The space $\ov{\mc M}{}_{g,n}^r$ is a Deligne--Mumford stack and has the structure of a compact complex orbifold. There is the natural forgetful map
\beqn
\ov{\mc M}{}_{g,n}^r \to \ov{\mc M}_{g,n}.
\eeqn

We recall certain facts about the local structure of the moduli space of stable $r$-spin curves. A more detailed discussion can be found in \cite[Section 2]{Tian_Xu_geometric_2}. 

We first describe the convergence within the same stratum of $\ov{\mc M}{}_{g,n}^r$. Let ${\mc C} = (\Sigma, \vec{\bf z}, L, \varphi)$ be a stable $r$-spin curve. By abuse of notations, let $(\Sigma, \vec{\bf z})$ denote the underlying coarse curve with markings. Let $j_{\mc C}$ be the complex structure on the underlying coarse curve $\Sigma$. Let ${\mc V}_{\rm def}$ be the tangent space of the moduli space of stable marked curves at the point represented by $(\Sigma, {\bf z})$. Choose a metric on ${\mc V}_{\rm def}$ whose $\epsilon$-ball is denoted by ${\mc V}_{\rm def}^\epsilon$. Then for a sufficiently small $\epsilon$ there exist a family of complex structures 
\beqn
j_\eta \in \Gamma( \Sigma, {\rm End}_{\mb R}(T \Sigma )),\ \eta \in {\mc V}_{\rm def}^\epsilon
\eeqn
such that
\begin{enumerate}
\item $j_0 = j_{\mc C}$; $j_\eta = j_{\mc C}$ in a neighborhood of markings and nodes for all $\eta\in {\mc V}_{\rm def}^\epsilon$; $j_\eta$ depends smoothly on $\eta$.

\item Every stable marked curve of the same combinatorial type and is sufficiently close to $(\Sigma, {\bf z})$ in $\ov{\mc M}_{g,n}$ is isomorphic to $(\Sigma_\eta, {\bf z})$ for some $\eta \in {\mc V}_{\rm def}^\epsilon$. 
\end{enumerate}
One can make the construction equivariant with respect to the automorphism group of the curves (see \cite[Section 2]{Tian_Xu_geometric_2}). However equivariance is not a concern of the current paper. 

The family of marked nodal surfaces $(\Sigma_\eta, {\bf z})$ inherit the orbifold structure from ${\mc C}$. Moreover, when $\eta$ is sufficiently close to $0 \in {\mc V}_{\rm def}$, the original $r$-spin structure also extends to $(\Sigma_\eta, {\bf z})$, giving a family of stable $r$-spin curves denoted by
\beqn
{\mc C}_{\eta, 0},\ \eta\in {\mc V}_{\rm def}^\epsilon.
\eeqn

Points near $[{\mc C}]$ in $\ov{\mc M}{}_{g, n}^r$ can be obtained by resolving the nodes of ${\mc C}_{\eta, 0}$. For each node $w \in {\mc C}$, one can choose orbifold chart
\beqn
N_w:=\{ (\xi_-, \xi_+) \in {\mb D} \times {\mb D}\ |\ \xi_- \xi_+ = 0\}/ {\mb Z}_{r_w}
\eeqn
(see \eqref{nodalchart}). Set
\beqn
{\mc V}_{\rm res} = {\mb C}^{w \in {\rm E}_{\mc C}}
\eeqn
which is the space of gluing parameters $\zeta = (\zeta_w)_{w \in \E_{\mc C}}$. For each $\zeta \in {\mc V}_{\rm res}$ sufficiently close to the origin, one can define a family of $r$-spin curves 
\beqn
{\mc C}_{\eta, \zeta},\ \eta \in {\mc V}_{\rm def},\ \zeta \in {\mc V}_{\rm res}
\eeqn
as follows. Inside $\Sigma_\eta$, for each node $w$, if we replace $N_w$ by the annulus
\beqn
N_{w, \zeta}:= \{ ( \xi_-, \xi_+)\in {\mb D} \times {\mb D}\ |\ \xi_- \xi_+ = \zeta_w \}/ {\mb Z}_{r_w}
\eeqn
which can be identified via the exponential map to the long cylinder
\beq\label{eqn24}
r_w \exp^{-1}(N_{w, \zeta}) \cong (r_w \log |\zeta_w|, 0) \times S^1,
\eeq
then we obtain a new curve $\Sigma_{\eta, \zeta}$.

The above construction provides a family of orbifold curves $(\Sigma_{\eta, \zeta}, \vec{\bf z}_{\eta, \zeta})$. The line bundles $L_{\eta, 0}$ also extends canonically to an $r$-th root of the log-canonical bundle by replacing the bundle chart over $N_w$ to a bundle chart over $N_{w, \zeta}$. Indeed, if $L_{\eta, 0}$ has a nodal chart of the form 
\beqn
\{ (\xi_-, \xi_+, t) \in {\mb D} \times {\mb D} \times {\mb C}\ |\ \xi_- \xi_+ = 0 \}/ {\mb Z}_{r_w}
\eeqn
then one can define $L_{\eta, \zeta} \to \Sigma_{\eta, \zeta}$ by replacing the above chart with the chart
\beqn
\{ (\xi_-, \xi_+, t) \in {\mb D} \times {\mb D} \times {\mb C}\ |\ \xi_- \xi_+ = \zeta_w \}/ {\mb Z}_{r_w}.
\eeqn
Denote the family of $r$-spin curves by ${\mc C}_{\eta, \zeta}$. 

We state the following lemma which can be viewed as a description of the local structure of the moduli space of stable $r$-spin curves. 

\begin{lemma}(cf. \cite{Tian_Xu_geometric_2})\label{domainconverge}
Let ${\mc C}_i$ be a sequence of smooth $r$-spin curves of genus $g$ with $n$ marked points and let ${\mc C}_\infty$ be a stable $r$-spin curve of genus $g$ with $n$ marked points. Then the sequence defined by the isomorphism class of ${\mc C}_i$ converges to the isomorphism class of ${\mc C}_\infty$ in $\ov{\mc M}{}_{g, n}^r$ if and only after removing finitely many elements in the sequence ${\mc C}_i$, there exist a sequence of deformation parameters $\eta_i$, a sequence of gluing parameters $\zeta_i$, and a sequence of isomorphisms of $r$-spin curves
\beqn
{\mc C}_i \cong {\mc C}_{\eta_i, \zeta_i}.
\eeqn
\end{lemma}

Formally, if we set $r = 1$, then Lemma \ref{domainconverge} reduces to the well-known result about universal unfoldings of stable marked curves (see for example \cite{Robbin_Salamon_2006}). 

\begin{notation}\label{notation27}
For any compact subset $Z \subset \Sigma_{\mc C}^*$ which is disjoint from the special points of ${\mc C}$, when $|\zeta|$ is sufficiently small, there is a canonical embedding $Z \hookrightarrow \Sigma_{{\mc C}_{\eta, \zeta}}^*$ and we denote the image by 
\beqn
Z_{\eta, \zeta} \subset \Sigma_{{\mc C}_{\eta, \zeta}}^*.
\eeqn
Moreover, the explicit construction of the family ${\mc C}_{\eta, \zeta}$ provides a family of long cylinders in $\Sigma_{{\mc C}_{\eta, \zeta}}$ (see \eqref{eqn24}). For each node $w$ of ${\mc C}$, denote the long cylinder \eqref{eqn24} by 
\beqn
r_w \exp^{-1} (N_{w, \zeta}) \cong (-T_{\eta, \zeta}(w), T_{\eta, \zeta} (w) ) \times S^1.
\eeqn
\end{notation}

\subsection{Cylindrical metrics and other structures}\label{subsection23}

In this subsection we discuss several structures for all $r$-spin curves that vary smoothly over the moduli space $\ov{\mc M}{}_{g, n}^r$. First we construct a family of conformal Riemannian metrics that are of cylindrical type over all $r$-spin curves. Let $g \geq 0$, $n \geq 0$ with $2g + n \geq 3$. For any smooth genus $g$ curve $\Sigma$ with $n$ marked points, denote the punctured Riemann surface by $\Sigma^*$. Then for any local holomorphic coordinate $z$ centered at a marked point $z_a$, the coordinate $s + {\bf i} t = - \log z$ is called a {\it cylindrical coordinate} on $\Sigma^*$ near the puncture $z_a$. The punctured neighborhood identified with a cylinder $[T, +\infty) \times S^1$ is called a {\it cylindrical end}.

We would like to specify certain type of Riemannian metrics on punctured Riemann surfaces that are of cylindrical type. We always take metrics whose conformal class belongs to the conformal class defined by the complex structure, hence a conformal Riemannian metric is equivalent to its area two-form.

\begin{defn}
An conformal Riemannian metric on $\Sigma^*$ is called a {\it cylindrical metric} of perimeter $2\pi$ if near each puncture $z_a$, there exists a cylindrical end $U_a \cong [S, +\infty) \times S^1$ with coordinate $s + {\bf i} t$ such that the area two-form is 
\beqn
\nu = \sigma(s, t) ds dt
\eeqn
where $\sigma: [S, +\infty)\times S^1 \to {\mb R}_+$ satisfying
\beq\label{eqn25}
\sup_{[S, +\infty) \times S^1} \Big[ e^{s} \big| \nabla^l (\sigma - 1) \big| \Big] \leq C_l,\ \forall l \geq 0.\footnote{The constants $C_l$ may not be bounded when $l$ grows.}
\eeq
\end{defn}
It is easy to verify that the above exponential decay condition is independent of the choice of the local holomorphic coordinate $z = e^{s + {\bf i} t}$.

\begin{defn}\label{defn29}
A family of {\it cylindrical metrics} over $\ov{\mc M}_{g,n}$ with $2g + n \geq 3$ is a collection of conformal Riemannian metrics $h_\Sigma$ on all $\Sigma^*$ where $\Sigma^*$ is the complement of special points of a stable $n$-marked genus $g$ Riemann surface, all of which satisfy the following conditions.

\begin{enumerate}

\item $h_\Sigma$ is a cylindrical metric (of perimeter $2\pi$) on each irreducible component.

\item If $\rho: (\Sigma_1, \vec{\bf z}_1 ) \to (\Sigma_2, \vec{\bf z}_2)$ is an isomorphism, then $\rho^* h_{\Sigma_2} = h_{\Sigma_1}$. 

\item Let $(\Sigma_{\eta_i, \zeta_i}, \vec{\bf z}_{\eta_i, \zeta_i}), i= 1, 2, \cdots$ be a sequence of curves converging to $(\Sigma, {\bf z})$ (cf. the $r=1$ case of Lemma \ref{domainconverge}). Let the complements of special points of them be $\Sigma_i^*$ and $\Sigma^*$ respectively. Then for any compact subset $Z \subset \Sigma^*$ disjoint from marked or nodal points, for large $i$, the restriction of $h_{\Sigma_i}$ to $Z_{\eta_i, \zeta_i} \cong Z$ converges smoothly to $h_\Sigma|_Z$ (see Notation \ref{notation27}). Moreover, for each node $w$ of $\Sigma$ with corresponding gluing parameter $\zeta_{i, w}$ nonzero, for any sequence of cylinders $[T_i -1, T_i + 1]\times S^1$ contained in the long cylinder $(-T_{\eta_i, \zeta_i}(w), T_{\eta_i, \zeta_i}(w) ) \times S^1 \subset \Sigma_{{\eta_i, \zeta_i}}$ with $|T_i \pm T_{\eta_i, \zeta_i}(w)| \to \infty$, the restriction of $h_{\Sigma_i}$ to $[T_i-1, T_i+1]\times S^1$ converges in $C^\infty$ to the standard cylindrical metric on $[-1, 1]\times S^1$. 
\end{enumerate} 
\end{defn}

There is a straightforward way of constructing such families of metrics. We refer to \cite{Venugopalan_quasi} for a proof. Cylindrical metrics are also essentially used by Fan--Jarvis--Ruan \cite{FJR_virtual} in the construction of FJRW invariants. 

\begin{prop}\cite[Proposition 1.4]{Venugopalan_quasi}
There exists a family of cylindrical metrics over $\ov{\mc M}_{g,n}$ for $g, n$ with $2g + n \geq 3$. 
\end{prop}

From now on, we fix the choice of such collection of cylindrical metrics (of perimeter $2\pi$) on stable genus $g$, $n$-marked Riemann surfaces. When we have a stable $r$-spin curve, we equip the underlying punctured Riemann surface the cylindrical metric we choose. 

The family of cylindrical metric can be viewed as a family of Hermitian metrics on the canonical bundle $K_{{\mc C}}$ of ${\mc C}$ restricted to the complements of the markings and nodes. However, in the concept of $r$-spin structures we used the notion of log-canonical bundles. Over each stable $r$-spin curve ${\mc C}$, one has an isomorphism 
\beq\label{eqn26}
K_{{\mc C}, {\rm log}}|_{\Sigma_{\mc C}^*} \cong K_{\Sigma_{\mc C}^*}
\eeq
because the divisor ${\mc O}(\vec{\bf z}):= {\mc O}(z_1) \otimes \cdots \otimes {\mc O}(z_n)$ is trivial away from the markings. However there is no canonical isomorphism. 

Consider the universal curve $\ov{\mc U}_{g, n} \to \ov{\mc M}_{g, n}$. The markings give sections $S_1, \ldots, S_n: \ov{\mc M}_{g, n} \to \ov{\mc U}_{g, n}$ whose images are divisors $D_1, \ldots, D_n\subset \ov{\mc U}_{g, n} \cong \ov{\mc M}_{g, n+1}$. Then one can choose a holomorphic section of the line bundle $[D_1 + \cdots + D_n]$ that vanishes exactly along these divisors. Then restrict to each fibre, this section provides a family of trivializations of ${\mc O}(\vec {\bf z})$ away from the markings which are invariant under automorphisms and vary in a holomorphic way over $\ov{\mc M}_{g, n}$. From now on we choose such a section, upon which the isomorphism \eqref{eqn26} is regarded as canonical. 

We claim that the cylindrical metrics and the canonical isomorphisms \eqref{eqn26} induce on each stable $r$-spin curve ${\mc C} = (\Sigma_{\mc C}, \vec{\bf z}_{\mc C}, L_{\mc C}, \varphi_{\mc C})$ a Hermitian metric on $L_{\mc C}$ away from punctures and nodes. Indeed, given a local holomorphic coordinate $z$ on $\Sigma_{\mc C}$ and a local there exists a local holomorphic section $e$ of $L_{\mc C}$ such that
\beqn
\varphi (e^{\otimes r}) = dz. 
\eeqn
Then define $\| e \| = (\| dz \|)^{\frac{1}{r}}$. Different choices of $e$ differ by an element of ${\mb Z}_r$ hence this is a well-defined metric on $L_{\mc C}$. Further, near each marking or node at which the monodromy of $L_{\mc C}$ is $e^{\frac{ 2\pi m {\bf i}}{r}}$, for any holomorphic coordinate $z$, there is a unique (up to ${\mb Z}_r$) section $e$ of $L_{\mc C}$ such that 
\beqn
\varphi (e^{\otimes r}) = \frac{dz}{z} z^m.
\eeqn
The metric near this marking is then
\beq\label{eqn27}
\| e \| = | z|^{m-1} ( \| dz \|)^{\frac{1}{r}} = |z|^m \left( \| \frac{ dz}{z} \| \right)^{\frac{1}{r}}.
\eeq
Namely, $\| e \|$ behaves like $|z|^m$ near the marking. 

Let $P_{\mc C} \to \Sigma_{\mc C}^*$ be the unit circle bundle of $L_{\mc C}$, which is a principal $U(1)$-bundle over the punctured surface. The Hermitian metric on $L_{\mc C}$ induces the Chern connection $A_{\mc C} \in {\mc A}(P_{\mc C})$. \eqref{eqn27} implies that at a puncture where $m \neq 0$, the connection is singular, and, under the trivialization induced by the unitary frame $e / \| e \|$, the connection is 
\beqn
A_{\mc C} = d + \frac{{\bf i} m}{r} dt + \alpha = d + {\bf i} q dt + \alpha
\eeqn
where $\alpha$ decays exponentially in the cylindrical coordinates, meaning that there exist $\delta>0$ and for all $l \geq 0$ a constant $C_l>0$ such that
\beqn
|\nabla^l \alpha| \leq C_l e^{- \delta s}.
\eeqn

Later we will use the following result, whose proof is left to the reader.

\begin{lemma}\label{lemma211}
For fixed $g, n$, given a family of cylindrical metrics over $\ov{\mc M}_{g, n}$ in the sense of Definition \ref{defn29}, there exists $C>0$ such that, for all stable genus $g$, $n$-marked $r$-spin curves ${\mc C}$, we have
\beqn
\| F_{A_{\mc C}}\|_{L^\infty(\Sigma_{\mc C}^*)} \leq C
\eeqn
where the norm is taken with respect to the cylindrical metric on $\Sigma_{\mc C}^*$. Moreover, we can choose the family of metrics for all $g, n$ such that for the same $C$ the above inequality is true for all $g, n$.
\end{lemma}

We also need to specify the universal structures over infinite cylinders. As discussed before, the $r$-spin structures over an infinite cylinder is classified by the monodromies of a line bundle $L_m \to \Thetait$ at $-\infty$. Let the cylindrical coordinate be $w = s + {\bf i} t$. Then there exists a global holomorphic sections of $L_m$ such that
\beqn
\varphi_m (e^{\otimes r}) = e^{mw} dw.
\eeqn
Denote $q = \frac{m}{r}$. Equip with $\Thetait$ the standard cylindrical metric of perimeter $2\pi$. Then define the Hermitian metric on $L_m$ by  $\| e \| \equiv |z|^q$. Then $\epsilon_m:= |z|^{-q} e$ is a global smooth section of unit length. Trivialize $L_m$ by $\epsilon_m$, we see that the Chern connection reads
\beq\label{eqn28}
A_{\mc C} = d + {\bf i} q dt. 
\eeq

\begin{lemma}\label{lemma212}
Let ${\mc C}_{\eta_i, \zeta_i}$ be a sequence of stable $r$-spin curves converging to ${\mc C}$. Remove finitely many elements from this sequence if necessary. Let $\Sigma_i^*$ and $\Sigma^*$ be their complements of special points, respectively. Then for any compact subset $Z \subset \Sigma^*$, which can be canonically identified with $Z_i \subset \Sigma_i^*$, there exist canonical smooth bundle isomorphisms 
\beqn
\vcenter{ \xymatrix{ P_{{\mc C}_{\eta_i, \zeta_i}} \ar[r]^{\cong} \ar[d] & P_{{\mc C}} \ar[d]\\
                     Z_i                            \ar[r]^{\cong}        & Z } }
\eeqn
with respect to which $A_{{\mc C}_{\eta_i, \zeta_i}}$ converges smoothly to $A_{\mc C}$. Moreover, for any sequence of cylinders $[T_i -1, T_i+1]\times S^1$ contained in the sequence of long cylinders $(-T_{\eta_i, \zeta_i}(w), T_{\eta_i, \zeta_i}(w) )\times S^1$ (see Notation \ref{notation27}) with $|T_i \pm T_{\eta_i, \zeta_i}(w)| \to \infty$, over which $P_{{\mc C}_{\eta_i, \zeta_i}}$ is trivialized by the coordinates, the restriction of $A_{{\mc C}_{\eta_i, \zeta_i}}$ to $[T_i-1, T_i + 1]\times S^1$ converges in $C^\infty$ to the connection \eqref{eqn28} over $[-1, 1]\times S^1$.
\end{lemma}

\section{The Gauged Witten Equation}\label{section3}

In this section we set up the gauged Witten equation. This equation is essentially induced from the GLSM Lagrangian \cite{Witten_LGCY} and can be viewed as the generalization of the vortex equation and the Witten equation. The current setting has appeared in \cite{Tian_Xu_2017}, while a slightly different version was used in \cite{Tian_Xu, Tian_Xu_2, Tian_Xu_3} by the authors.

\subsection{The GLSM space}

\begin{defn}
A {\bf GLSM space} is a quadruple $(V, G, W, \mu)$ where
\begin{enumerate}

\item $V$ is a K\"ahler manifold with a holomorphic ${\mb C}^*$-action (called the {\it R-symmetry}). 

\item $G$ is a connected\footnote{Our theory can be extended to the case when $G$ is disconnected. The extension is routine but needs more involved treatment of nontrivial bundles over cylinders.} complex reductive Lie group acting holomorphically on $V$ such that: a) the $G$-action commutes with the R-symmetry; b) the restriction of the $G$-action to its maximal compact subgroup $K\subset G$ is Hamiltonian. 

\item $W: V \to {\mb C}$ is a $G$-invariant holomorphic function and is homogeneous of degree $r \geq 1$ with respect to the R-symmetry, namely
\beqn
W(\xi x) = \xi^r W(x),\ \forall \xi \in {\mb C}^*,\ x\in V.
\eeqn

\item $\mu: V \to {\mf k}^*$ is a moment map for the $K$-action with $0$ being a regular value.
\end{enumerate}
\end{defn}

Since the $G$-action commutes with the ${\mb C}^*$-action, they induce an action by $\ubar G:= G \times {\mb C}^*$ whose maximal compact subgroup is $\ubar K:= K \times U(1)$. For any vector $\xi$ in ${\mf k}$, $\mf{u}(1)$, or the Lie algebra $\ubar {\mf k}$ of $\ubar K$, let 
\beqn
{\mc X}_\xi \in \Gammait(TV)
\eeqn
denote the infinitesimal action of $\xi$. 

The following shows an important property of $W$.
\begin{lemma}
Any critical value of $W$ must be zero.
\end{lemma}

\begin{proof}
By the homogeneity of $W$ with respect to the R-symmetry, at a critical point $x$ of $W$, we have
\beqn
0 = {\mc X}_{\eta} W(x) = r \eta  W(x).
\eeqn
where ${\mc X}_{\eta}$ is the infinitesimal ${\mb C}^*$-action for $\eta \in {\rm Lie}{\mb C}^*$. Hence $W(x) = 0$.
\end{proof}

Given a GLSM space $(V, G, W, \mu)$, define the semi-stable locus of $V$ as 
\beqn
V^{\rm ss}:= G \mu^{-1}(0) \subset V.
\eeqn
Denote by ${\rm Crit}W\subset V$ the critical locus of $W$ and denote
\beqn
({\rm Crit} W)^{\rm ss}:= {\rm Crit} W \cap V^{\rm ss}.
\eeqn

\begin{defn}\label{defn33}
A GLSM space $(V, G, W, \mu)$ is in the {\bf geometric phase} if $W$ is Morse--Bott in $V^{\rm ss}$, i.e., if $dW|_{V^{\rm ss}}$ intersects cleanly with the zero section of $T^*V $.
\end{defn}

For the rest of this paper, we will assume the following technical assumptions.
\begin{hyp}\label{hyp34}
\begin{enumerate}

\item There is a homomorphism $\iota_W: {\mb C}^* \to G$ such that
\beqn
\xi \cdot x = \iota_W (\xi) \cdot x,\ \forall x\in \ov{({\rm Crit} W)^{\rm ss}}.
\eeqn

\item The action of $e^{\frac{2\pi \i}{r}} \iota_W( e^{-\frac{2\pi\i}{r}}) \in U(1) \times K$ on $V$ is the identity.

\item The $K$-action on $\mu^{-1}(0)$ is free. 

\item The restriction $\mu$ to $\ov{({\rm Crit} W)^{\rm ss}}$ is proper. 

\item $V$ is symplectically aspherical: for any smooth map $f: S^2 \to V$ there holds 
\beqn
\int_{S^2} f^* \omega_V = 0.
\eeqn 

\end{enumerate}
\end{hyp}

We also have the following assumption on the geometry at infinity. 

\begin{hyp}\label{hyp35}
There exists $\xi_W$ in the center $Z({\mf k})\subset {\mf k}$, and a continuous function $\tau \mapsto c_W (\tau)$ (for $\tau \in Z({\mf k})$) satisfying the following condition. If we define ${\mc F}_W:= \mu\cdot \xi_W$, then ${\mc F}_W|_{\ov{({\rm Crit} W)^{\rm ss}}}: \ov{({\rm Crit} W)^{\rm ss}} \to {\mb R}$ is proper, bounded from below, and
\beq\label{eqn31}
\begin{array}{c} x \in \ov{({\rm Crit} W)^{\rm ss}}\\
\xi \in T_x V\\
  {\mc F}_W(x) \geq c_W(\tau)
\end{array} \Longrightarrow \left\{ \begin{array}{c} \langle \nabla_\xi \nabla {\mc F}_W(x), \xi\rangle + \langle \nabla_{J \xi} \nabla {\mc F}_W(x), J \xi \rangle \geq 0, \\
\langle \nabla {\mc F}_W(x), J {\mc X}_{\mu(x) - \tau}(x) \rangle \geq 0. \end{array}\right. 
\eeq
\end{hyp}

\begin{rem}
We explain why we need these assumptions.
\begin{enumerate}
    \item Hypothesis 3.4 (a) seems to be the most mysterious assumption. It says that inside the semi-stable locus the R-symmetry can be merged into the gauge symmetry. This condition is satisfied by many interesting examples such as the most classical example of projective hypersurfaces recalled in the next subsection. We also notice that the map $\iota_W$ depends on the stability condition. This means that if we change the stability condition, then we shall require Condition 3.4 (a) with a different group homomorphism $\iota_W$. On the other hand, this hypothesis implies that the solution to the gauged Witten equation are actually symplectic vortices (see Remark \ref{rem46}), which is a very useful fact. 

    \item Hypothesis 3.4 (b) is assumed in order to guarantee that the degree of solutions (see the definition of the degree after Theorem \ref{thm42}) has integer coefficients. Otherwise the degree of solutions may have fractional coefficients and certain arguments, such as the index formula, need to be modified accordingly.

    \item Hypothesis 3.4 (c) implies that the classical vacuum is a manifold. We could allow the action to be only locally free, resulting in an orbifold classical vacuum. However this will significantly complicate the construction.

    \item Hypothesis 3.4 (d), together with Hypothesis 3.5, are needed for proving the $C^0$-compactness of solutions (see the proof of Theorem \ref{thm43}). Essentially we need certain maximal principle to guarantee that solutions won't escape at infinity. Her   e Hypothesis \ref{hyp35} is similar to an assumption in \cite{Cieliebak_Gaio_Salamon_2000}\cite{Cieliebak_Gaio_Mundet_Salamon_2002}.

    \item Hypothesis 3.4 (e) is assumed to simplify the technicality. If we remove this assumption, then the compactification of the moduli spaces will require considering sphere bubbles. That could be useful, for example, when the target space $V$ is some vector bundle over a compact manifold.
    \end{enumerate}
\end{rem}

Hypothesis \ref{hyp34} above allow us to consider the symplectic reductions
\begin{align*}
&\ V \qu G:= \mu^{-1}(0)/K \cong V^{\rm ss}/G, &\ X:= ({\rm Crit} W \cap \mu^{-1}(0))/K \cong ( {\rm Crit} W)^{\rm ss} /G.
\end{align*}
Our hypothesis implies that $V \qu G$ is a K\"ahler manifold and $X \subset V \qu G$ is a closed submanifold. The manifold $X$ is called the {\bf classical vacuum} of $(V, G, W, \mu)$. The third item of Hypothesis \ref{hyp34} implies that $X$ is compact.

\subsection{Example: projective hypersurfaces}\label{subsection32}

Consider a degree $r\geq 2$ homogeneous polynomial $Q: {\mb C}^N \to {\mb C}$. It is called nondegenerate if $0\in {\mb C}^N$ is the only critical point of $Q$, implying that the hypersurface $X_Q \subset {\mb P}^{N-1}$ defined by $Q$ is a smooth manifold.

To realize the target $X_Q$ from a GLSM space, following Witten \cite{Witten_LGCY}, consider $V = {\mb C}^{N+1}$ with coordinates $(p, x_1, \ldots, x_N)$. The R-symmetry is the standard ${\mb C}^*$-action on the $x$-variables. Introduce $W: {\mb C}^{N+1}\to {\mb C}$ which is defined as
\beqn
W(p, x_1, \ldots, x_N) = p Q(x_1, \ldots, x_N). 
\eeqn
Let $G$ be ${\mb C}^*$ and act on $X$ by  
\beqn
g \cdot (p, x_1, \ldots, x_N) = ( g^{-r} p, g  x_1, \ldots, g x_N).
\eeqn
It is Hamiltonian and a moment map is 
\beqn
\mu (p, x_1, \ldots, x_N) = - \frac{\i}{2} \Big[ |x_1|^2 + \cdots + |x_N|^2 - r |p|^2 - \tau \Big]
\eeqn
where $\tau$ is a constant, playing the role as a parameter of this theory. 

One can check that when $\tau >0$, the quadruple $(V, G, W, \mu)$ is a GLSM space in geometric phase satisfying all of our technical assumptions. Indeed, the critical locus of $W$ decomposes as 
\beqn
{\rm Crit} W = \Big\{ (0, x_1, \ldots, x_N)\ |\ Q(x_1, \ldots, x_N) = 0 \Big\} \cup \Big\{ (p, 0, \ldots, 0) \ |\ p \in {\mb C} \Big\}.
\eeqn
We see that the first component above contains the smooth locus, and when $\tau>0$, $\mu^{-1}(0)$ is disjoint from the second component above. Hence 
\beqn
\ov{({\rm Crit} W)^{\rm ss}} = \Big\{ (0, x_1, \ldots, x_N)\ |\ Q(x_1, \ldots x_N) = 0\Big\} \cong Q^{-1}(0) \subset {\mb C}^N.
\eeqn
and the classical vacuum $X$ is the projective hypersurface defined by $Q$. Hypothesis \ref{hyp34} and Hypothesis \ref{hyp35} can be checked in a straightforward manner. 

We remark that when $\tau<0$, we are in the Landau--Ginzburg phase. Witten's argument for LG/CY correspondence \cite{Witten_LGCY} is based on this observation.

\subsection{Gauged maps and vortices}\label{subsection33}

We review the basic notions about gauged maps from surfaces to a manifold with group action. Let $K$ be a compact Lie group. Let $\Sigma$ be a surface and $P \to K$ be a smooth principal $K$-bundle. Let $M$ be a smooth manifold acted by $K$. Let $\pi: P(M) \to \Sigma$ be the associated fibre bundle with fibre $M$. The vertical tangent bundle is then  
\beqn
P(TM) \to P(M).
\eeqn
Any connection $A \in {\mc A}(P)$ induces a splitting of
\beqn
\xymatrix{ 0 \ar[r]  & P(TM) \ar[r]  & TP(M) \ar[r] & \pi^* T\Sigma \ar[r] & 0 }.
\eeqn

A {\it gauged map} from $\Sigma$ to $M$ is a triple ${\bm v} = (P, A, u)$ where $P \to \Sigma$ is a principal $K$-bundle, $A \in {\mc A}(P)$ is a connection and $u$ is a section of $P(M)$. When $P$ is fixed in the context, we also call ${\bm v} = (A, u)$ a gauged map. If $u$ is (weakly) differentiable, one has the ordinary derivative 
\beqn
du \in \Gamma(\Sigma, \Lambda^1 \otimes u^* TP(M)). 
\eeqn
Using the projection $TP(M) \to P(TM)$ induced by $A$, one has the {\it covariant derivative} 
\beqn
d_A u  \in \Gamma( \Sigma, \Lambda^1 \otimes u^* P(TM)) = \Gamma( \Sigma, {\rm Hom}( T\Sigma, u^* P(TM)). 
\eeqn

If $M$ has a $K$-invariant almost complex structure $J_M$, then it induces a complex structure on $P(TM)$. Meanwhile if $\Sigma$ has a complex structure, then one can take the $(0, 1)$-part of the covariant derivative, which is denoted by
\beqn
\ov\partial_A u  \in \Omega^{0,1} ( u^* P(TM)) \cong \Gamma(\Sigma, {\rm Hom}^{0,1}(T\Sigma, u^* P(TM)).
\eeqn

Associated to the adjoint representation there is the adjoint bundle ${\rm ad} P = P({\mf k})$ of $P$. The space connections ${\mc A}(P)$ is an affine space modelled on the linear space $\Omega^1(\Sigma, {\rm ad}P)$. The curvature of $A$ is a 2-form 
\beqn
F_A \in \Omega^2 (\Sigma, {\rm ad} P). 
\eeqn

A smooth {\it gauge transformation} on $P$ is a smooth map $g: P \to K$ satisfying 
\beqn
g(p h) = h^{-1} g(p) h,\ \forall h \in K. 
\eeqn
Here $ph$ denotes the canonical right $K$-action on the principal bundle $P$. Gauge transformations naturally form an infinite-dimensional Lie group ${\mc G}(P)$, using the multiplication of $G$ and they can be viewed as automorphisms of $P$ (i.e., smooth fibre bundle automorphisms that respect the right $K$-action). It acts on connections of $P$ and sections of $P(M)$ by reparametrization. Hence we regard gauge transformations as right actions by ${\mc G}(P)$, denoted by 
\beqn
g^* (A, u) = (g^* A, g^* u).
\eeqn
The covariant derivatives and curvatures transform naturally with respect to gauge transformations. Namely, for any $A \in {\mc A}(P)$, $u \in \Gamma(P(M))$, and $g \in {\mc G}(P)$, 
\begin{align*}
&\ d_{g^* A} g^* u = g^{-1} d_A u,\ &\ F_{g^* A} = {\rm Ad}_g^{-1} F_A. 
\end{align*}

Now assume that $M$ has a symplectic form $\omega_M$ and the $K$-action is Hamiltonian with moment map $\mu: M \to {\mf k}^*$. Then $\mu$ induces an element 
\beqn
\mu_{P(M)} \in \Gamma ( P(M), {\rm ad} P^*)
\eeqn
where the (non-standard) notation means smooth maps between fibre bundles over $\Sigma$. Then if $u$ is a section of $P(M)\to \Sigma$, one can define the composition 
\beqn
\mu(u) \in \Gamma(\Sigma, {\rm ad} P^*).
\eeqn

Now we can review the basic set up of the symplectic vortex equation. This is an equation introduced first by Mundet \cite{Mundet_thesis, Mundet_2003} and Cieliebak--Gaio--Salamon \cite{Cieliebak_Gaio_Salamon_2000}. Assume that the $K$-action is Hamiltonian with moment map $\mu: M \to {\mf k}^* $ and there is a $K$-invariant almost complex structure $J_M$. Choose an adjoint-invariant metric on ${\mf k}$ so that we can identify ${\mf k} \cong {\mf k}^*$. Choose an area form on $\Sigma$ so we can identify a zero-form with a two-form on $\Sigma$ by the Hodge star operator $*$. A {\it vortex} is a gauged map $(P, A ,u)$ from $\Sigma$ to $M$ satisfying the following equation 
\begin{align}\label{eqn32}
&\ \ov\partial_A u = 0,\ &\ * F_A + \mu(u) = 0.
\end{align}
Here the second equation is viewed as an equality in the space $\Gamma(\Sigma, {\rm ad} P)$. 

The vortex equation can be viewed as the {\it equation of motion} with respect to the above energy functional. To see we need to have a short review of equivariant topology. The equivariant (co)homology of $M$ is defined to be the usual (co)homology of the Borel construction 
\beqn
M_K = EK \times_K M,
\eeqn
where $EK \to BK$ is the universal bundle over the classifying space $BK$. Then for any gauged map $(P, A, u)$ over $\Sigma$, $P$ induces a homotopy class of maps of $K$-bundles
\beqn
\vcenter{ \xymatrix{ P \ar[r] \ar[d] & EK \ar[d] \\
           \Sigma   \ar[r] & BK } }.
\eeqn
The section $u$ is equivalent to an equivariant map $\Phi_u: P \to M$. Hence there is a well-defined homotopy class of equivariant maps $P \to EK \times M$. It descends to a map $\Sigma \to M_K$. If $\Sigma$ is closed, the section gives a homology class $B \in H_2^K(M; {\mb Z})$. 

For a solution $(P, A, u)$ to the vortex equation \eqref{eqn32} over a compact surface $\Sigma$, suppose $(P, u)$ represents a class $B \in H_2^K(M; {\mb Z})$, then (see for example \cite{Cieliebak_Gaio_Mundet_Salamon_2002})
\beqn
E(P, A, u):= \frac{1}{2} \left( \| d_A u \|_{L^2}^2 + \| F_A \|_{L^2}^2 + \| \mu(u)\|_{L^2}^2 \right) = \langle [\omega_M^K], B\rangle.
\eeqn
Here $[\omega_M^K] \in H^2_K(M)$ is the equivariant cohomology class represented by the equivariant form $\omega_M + \mu$.

\subsection{Gauged Witten equation}\label{subsection34}

Let ${\mc C}= (\Sigma_{\mc C}, L_{\mc C}, \varphi_{\mc C}, \vec{\bf z}_{\mc C})$ be an $r$-spin curve and let $\Sigma_{\mc C}^* \subset \Sigma_{\mc C}$ be the complement of the markings and nodes, which is a smooth open Riemann surface. 
Recall that in Subsection \ref{subsection23} we have chosen a Hermitian metric on $L_{\mc C}$ which only depends on the isomorphism class of ${\mc C}$. Then $P_{\mc C} \to \Sigma_{\mc C}^*$ is the associated $U(1)$-bundle and $A_{\mc C} \in {\mc A}(P_{\mc C})$ is the Chern connection. If $P \to \Sigma_{\mc C}^*$ is a principal $K$-bundle, then we denote by
\beqn
\ubar P_{\mc C} \to \Sigma_{\mc C}^*
\eeqn
the principal bundle with structure group $K\times U(1)$ obtained from $P$ and $P_{\mc C}$. When ${\mc C}$ is fixed from the context, we will abbreviate $\ubar P_{\mc C}$ by $\ubar P$. Since the $K$-action and the $U(1)$-action on $V$ commute, one has the associated bundle 
\beqn
\ubar P (V):= (\ubar P \times V)/ K \times U(1).
\eeqn
Moreover, together with $A_{\mc C}$, any connection $A$ on $P$ induces a connection $\ubar A$ on $\ubar P$. 

\begin{defn}
A {\it gauged map} from ${\mc C}$ to $V$ is a triple ${\bm v} = (P, A, u)$ where $P \to \Sigma_{\mc C}^*$ is a $K$-bundle, $A \in {\mc A}(P)$ is a connection, and $u \in \Gamma ( \ubar P(V))$ is a section.
\end{defn}

Fix $P \to \Sigma_{\mc C}^*$. We can lift the superpotential $W : V \to {\mb C}$ to a section 
\beqn
{\mc W}_{\mc C} \in \Gamma( \ubar P(V), \pi^* K_{{\mc C}, {\rm log}}) \cong \Gamma( \ubar P (V), \pi^* K_{\Sigma_{\mc C}^*})
\eeqn
as follows. A point of $\ubar P(V)$ is represented by a triple $[p_{\mc C}, p, x]$ where $p_{\mc C } \in P_{\mc C}$, $p \in P$, and $x \in V$, with the equivalence relation defined by
\beqn
[p_{\mc C} h_{\mc C}, p h, x ] = [ p_{\mc C}, p, h_{\mc C} h x ],\ \forall h_{\mc C} \in U(1), h \in K.
\eeqn
Then define 
\beqn
{\mc W}_{\mc C} ([p_{\mc C}, p, x]) = W(x) \varphi_{\mc C} ( (p_{\mc C})^r).
\eeqn
By the homogeneity and $K$-invariance of $W$, ${\mc W}_{\mc C}$ is a well-defined section. 

Since the K\"ahler metric is $K \times U(1)$-invariant, it induces a Hermitian metric on the vertical tangent bundle $\ubar P (TV) \to \ubar P (V)$. Therefore, one can dualize the differential $d{\mc W}_{\mc C} \in \Gamma( \ubar P(V), \pi^*K_{\Sigma_{\mc C}^*} \otimes \ubar P (TV)^*)$, obtaining the gradient
\beqn
\nabla {\mc W}_{\mc C} \in \Gamma( \ubar P(V), \pi^* \ov{ K_{\Sigma_{\mc C}^*}} \otimes \ubar P (TV)).
\eeqn
On the other hand, for each $A \in {\mc A}(P)$ and $u \in \Gamma( \ubar P (V))$, one can take the covariant derivative of $u$ with respect to $\ubar A$, and take its $(0,1)$ part. We regard this as an operation on the pair $(A, u)$ and hence denote the covariant derivative by $d_A u$ and its $(0,1)$ part as 
\beqn
\ov\partial_A u \in \Gamma( \Sigma_{\mc C}^*, \ov{K_{\Sigma_{\mc C}^*} } \otimes u^* \ubar P(TV))=: \Omega^{0,1}(\Sigma_{\mc C}^*, u^* \ubar P  (TV) ).
\eeqn
We have the {\it Witten equation}
\beq\label{eqn37}
\ov\partial_A u + \nabla {\mc W}_{\mc C} (u) = 0.
\eeq

We need to fixed the complex gauge by imposing a curvature condition. Let ${\rm ad} P$ be the adjoint bundle. Then the curvature form $F_A$ of $A$ is 2-form with coefficients in ${\rm ad} P $. Further, since the R-symmetry commutes with the $K$-action, for any $u \in \Gamma( \ubar P (V))$, $\mu (u)$ is a well-defined section of ${\rm ad} P $. The {\it vortex equation} is 
\beq\label{eqn38}
* F_A + \mu (u) = 0.
\eeq
The {\bf gauged Witten equation} over ${\mc C}$ is the system on $(P, A, u)$
\begin{align}\label{eqn39}
&\ \ov\partial_A u + \nabla {\mc W}_{\mc C}(u) = 0,\ &\ * F_A + \mu (u) = 0.
\end{align}

Basic local properties of solutions to the gauged Witten equation still hold as for vortices. For example, there is a gauge symmetry of \eqref{eqn39} for gauge transformations $g: P \to K$. Moreover, for any weak solution, there exists a gauge transformation making it a smooth solution.

Fix the smooth $r$-spin curve ${\mc C}$. For a gauged map ${\bm v} = (P, A, u)$ from ${\mc C}$ to $X$, define its {\it energy} as
\beqn
E(P, A, u) =  \frac{1}{2} \Big( \| d_A u \|_{L^2}^2 + \| \mu  (u) \|_{L^2}^2 + \| F_A \|_{L^2}^2 \Big) + \| \nabla {\mc W}_{\mc C} (u)\|_{L^2}^2.
\eeqn

\begin{defn}
Let ${\mc C}$ be a smooth $r$-spin curve. We say that a gauged map $(P, A, u)$ from ${\mc C}$ to $V$ is {\it bounded} if it has finite energy and if there is a $\ubar K$-invariant compact subset $Z \subset V$ such that 
\beqn
u(\Sigma_{\mc C}^* ) \subset \ubar P (Z) \subset \ubar P (V). 
\eeqn
\end{defn}

\section{Analytical Properties of Solutions}\label{section4}

In this section we collect several properties of solutions to the gauged Witten equation over a smooth $r$-spin curve. Some of the results are proved in the companion paper \cite{Tian_Xu_geometric_2}. A consequence of these results is that solutions represents certain equivariant homology classes in $V$ in integer coefficients. We will also prove a uniform $C^0$-bound of solutions, a prerequisite of compactifying the moduli space. We remind the reader that our technical assumptions (Hypothesis \ref{hyp34} and Hypothesis \ref{hyp35}) are important for deriving these properties.

We first have the following result showing that solutions are contained in the critical locus and hence are all holomorphic. It is proved in \cite{Tian_Xu_geometric_2}.

\begin{thm}\label{thm41}\cite{Tian_Xu_geometric_2}
Let ${\bm v} = (P, A, u)$ be a bounded solution to \eqref{eqn39} over a smooth $r$-spin curve ${\mc C}$.  Then
\beqn
\ov\partial_A u \equiv \nabla {\mc W}_{\mc C} (u) \equiv 0.
\eeqn
\end{thm}

Another expected result is about the asymptotic behavior of solutions over the cylindrical ends. Its proof is also deferred to \cite{Tian_Xu_geometric_2}. We first introduce a few notations. Given $m \in \{0, 1, \ldots, r-1\}$ which can label the monodromy of $r$-spin structures, denote 
\beqn
\lambda_m:= - \iota_W( \frac{{\bf i} m}{r} ) \in {\mf k}. 
\eeqn

\begin{thm}\label{thm42}
Let ${\bm v} = (P, A, u)$ be a bounded solution of the gauged Witten equation over a smooth $r$-spin curve ${\mc C}$. Suppose the monodromy of the $r$-spin structure at the puncture $z_a$ is $m_a \in \{0, 1, \ldots, r-1\}$. Then there exist $x_a \in {\rm Crit} W \cap \mu^{-1}(0)$ and a smooth trivialization of $P$ near $z_a$ such that, with respect to this trivialization, we can write $A = d + \phi ds + \psi dt$ and regard $u$ as a map $u_a: U_a \to V$, satisfying 
\beq\label{eqn41}
\lim_{s \to +\infty} \phi = 0,\ \lim_{s \to +\infty} \psi = \lambda_{m_a},\ \lim_{s \to +\infty} u(s, t) = x_a. 
\eeq
\end{thm}

Theorem \ref{thm41} and \ref{thm42} imply that if ${\mc C}$ has at least one marking, then a solution is contained in the closure of the semi-stable locus $\ov{({\rm Crit}W)^{\rm ss}}$. A further consequence is that a bounded solution represents a $\ubar K$-equivariant homology class in $V$ (usually called the degree). The degree of the solution is defined as follows. The asymptotic behavior of solutions shown in Theorem \ref{thm42} specifies particular trivializations of the $K$-bundle $P\to \Sigma_{\mc C}^*$ over cylindrical ends. Together with the trivialization of $P_{\mc C}$ over the cylindrical ends the bundle $\ubar P$ extends to a $K \times U(1)$-bundle over the underlying smooth (not orbifold) surface, which is still denoted by $\ubar P$. Theorem \ref{thm42} and item (b) of Hypothesis \ref{hyp34} imply that the section $u$ extends to a continuous section of $\ubar P (V)$. Hence by the discussion at the end of Subsection \ref{subsection33}, there is an integer class 
\beqn
\ubar B\in H_2^{\ubar K} (V; {\mb Z})
\eeqn
associated to the homotopy class of the pair $(\ubar P, u)$. This class is defined to be the degree of the solution. The class $\ubar B$ can be paired with the equivariant symplectic class $[\omega_V^{\ubar K}]$ which will appear in the energy computation; it can also be paired with the equivariant first Chern class $c_1^{\ubar K}(TV)$ which will appear in the dimension axiom of the virtual cycle. 

\subsection{Uniform $C^0$ bound}

The next major result is the uniform $C^0$ bound on solutions. This is a crucial step towards the compactness of moduli spaces of the gauged Witten equation.

\begin{thm}\label{thm43}
For each positive number $E>0$, there exists a $\ubar K$-invariant compact subset $N = N_E \subset V$. Given a smooth $r$-spin curve ${\mc C} = (\Sigma_{\mc C}, \vec{\bf z}_{\mc C}, L_{\mc C}, \varphi_{\mc C})$ and a solution ${\bm v} = (P, A, u)$ to the gauged Witten equation \eqref{eqn39} with $E({\bm v}) \leq E$ and image contained in $\ubar P(\ov{({\rm Crit} W)^{\rm ss}})$. Then ${\rm Im} (u) \subset \ubar P (N)$. Moreover, if ${\mc C}$ has at least one marked point, then $N_E$ can be chosen to be independent of $E$. 
\end{thm}

\begin{proof}
Consider the function ${\mc F}_W = \mu \cdot \xi_W$ given by Hypothesis \ref{hyp35}. Choose local coordinate $z = s + {\bf i} t$ on $\Sigma_{\mc C}^*$. Let $\Delta$ be the standard Laplacian in this coordinate. Let the volume form be $\sigma ds dt$ and let the curvature form of $A_{\mc C}$ be $F_{A_{\mc C}} = \kappa_{\mc C} ds dt$. We also choose local trivializations of $P_{\mc C}$ and $P$ so that the combined connection can be written as $d + \ubar \phi ds + \ubar \psi dt$ and the section is identified with a map $u$ into $V$. 

The covariant derivatives of a tangent vector field $\xi$ along $u$ are
\begin{align*}
&\ D_{\ubar A, s} \xi = \nabla_s \xi + \nabla_\xi {\mc X}_{\ubar \phi},\ &\ D_{\ubar A, t} \xi = \nabla_t \xi + \nabla_\xi {\mc X}_{\ubar \psi}.
\end{align*} 
Then by the vortex equation $*F_A + \mu(u) = 0$, one has	
\beqn
\begin{split}
D_{\ubar A, s} {\bm v}_t - D_{\ubar A, t} {\bm v}_s = &\ \nabla_s (\partial_t u + {\mc X}_{\ubar \psi} ) + \nabla_{\partial_t u + {\mc X}_{\ubar \psi} } {\mc X}_{\ubar \phi}  - \nabla_t (\partial_s u + {\mc X}_{\ubar \phi} ) - \nabla_{\partial_s u + {\mc X}_{\ubar \phi} } {\mc X}_{\ubar \psi} \\
                                = &\ \nabla_{\partial_s u} {\mc X}_{ \ubar \psi} + {\mc X}_{\partial_s \ubar \psi} + \nabla_{\partial_t u + {\mc X}_{\ubar \psi} } {\mc X}_{\ubar \phi } - \nabla_{\partial_t u} {\mc X}_{ \ubar \phi} - {\mc X}_{\partial_t \ubar \phi} - \nabla_{\partial_s u + {\mc X}_{\ubar \phi} } {\mc X}_{\ubar \psi} \\
																= &\ {\mc X}_{\partial_s \ubar \psi} - {\mc X}_{\partial_t \ubar \phi } + [{\mc X}_{\ubar \psi} , {\mc X}_{\ubar \phi} ]\\
																= &\ {\mc X}_{\kappa_{\mc C}} -  \sigma {\mc X}_{\mu}.
\end{split}
\eeqn
Moreover, since $u$ is holomorphic with respect to $\ubar A$, ${\mc F}_W$ is $\ubar K$-invariant, and $D_{\ubar A, s}, D_{\ubar A, t}$ are metric-preserving, one has
\beqn
\begin{split}
&\ \Delta {\mc F}_W (u) \\
 = &\ \partial_s \langle \nabla {\mc F}_W (u), \partial_s u \rangle + \partial_t \langle \nabla {\mc F}_W (u), \partial_t u \rangle\\
                     = &\ \partial_s \langle \nabla {\mc F}_W (u), {\bm v}_s \rangle + \partial_t \langle \nabla {\mc F}_W (u), {\bm v}_t \rangle\\
                     = &\ \langle D_{\ubar A, s} \nabla {\mc F}_W(u), {\bm v}_s \rangle + \langle \nabla{\mc F}_W(u), D_{\ubar A, s}  {\bm v}_s \rangle + \langle D_{\ubar A, t} \nabla {\mc F}_W(u), {\bm v}_t \rangle + \langle \nabla {\mc F}_W(u), D_{\ubar A, t} {\bm v}_t \rangle \\
		            = &\ \langle \nabla_{{\bm v}_s} \nabla {\mc F}_W, {\bm v}_s \rangle + \langle \nabla_{{\bm v}_t} \nabla {\mc F}_W, {\bm v}_t \rangle + \langle \nabla {\mc F}_W, - J D_{\ubar A, s} {\bm v}_t + J D_{\ubar A, t} {\bm v}_s \rangle \\
		            = &\ \langle \nabla_{{\bm v}_s} \nabla {\mc F}_W, {\bm v}_s \rangle + \langle \nabla_{{\bm v}_t} \nabla {\mc F}_W, {\bm v}_t \rangle +  \langle \nabla {\mc F}_W, -J {\mc X}_{\kappa_{\mc C}} +  \sigma J {\mc X}_{\mu}\rangle.
										\end{split}
\eeqn
Since $u$ is contained in $\ov{({\rm Crit} W)^{\rm ss}}$, by Hypothesis \ref{hyp34}, one has
\beqn
\Delta {\mc F}_W(u)  = \langle \nabla_{{\bm v}_s} \nabla {\mc F}_W, {\bm v}_s \rangle + \langle \nabla_{{\bm v}_t} \nabla {\mc F}_W, {\bm v}_t \rangle + \langle \nabla {\mc F}_W, - J {\mc X}_{\iota_W(\kappa_{\mc C})} +  \sigma {\mc X}_{\mu} \rangle.
\eeqn
By Lemma \ref{lemma211}, which says that $F_{A_{\mc C}} = \kappa_{\mc C} ds dt$ is uniformly bounded, one can cover $\Sigma_{\mc C}^*$ by coordinate charts such that for certain $\tau_R>0$ independent of ${\mc C}$ and ${\bm v}$, one always have 
\beqn
|\kappa_{\mc C} | \leq \tau_R.
\eeqn
Notice that the image of $\iota_W$ in $Z({\mf k})$ is one-dimensional. By Hypothesis \ref{hyp35} and \eqref{eqn31}, whenever ${\mc F}_W(u) \geq \max \big\{ c_W( \tau_R {\bf i} ), c_W(-\tau_R {\bf i}) \big\}$, one has
\beqn
\Delta {\mc F}_W(u) \geq 0.
\eeqn
Therefore, ${\mc F}_W(u)$ is subharmonic whenever it is greater than $\max \big\{ c_W( \tau_R {\bf i}), c_W(-\tau_R {\bf i}) \big\}$. Assume that
\beqn
C:= \sup_{{\Sigma_{\mc C}^*}} {\mc F}_W(u) > \max \big\{ c_W( \tau_R {\bf i}), c_W(-\tau_R {\bf i}) \big\}.
\eeqn
If ${\mc C}$ has at least one puncture, then the subharmonicity and the asymptotic behavior of solutions imply that 
\beqn
C \leq \sup_{x \in {\rm Crit} W\cap \mu^{-1}(0)} {\mc F}_W(x).
\eeqn
If ${\mc C}$ has no puncture, then ${\mc F}_W(u)\equiv C$ over $\Sigma_{\mc C}^* = \Sigma_{\mc C}$, whose area is finite. By the vortex equation, one has
\beqn
0 = * F_{A} \cdot \xi_W + \mu(u)\cdot \xi_W.
\eeqn
Integrate this equality over $\Sigma_{\mc C}$, one obtains
\beqn
C {\rm Area} \Sigma_{\mc C} \leq   |\xi_W| \int_{\Sigma_{\mc C}} |F_A| \sigma ds dt
 \leq  |\xi_W| \sqrt{ {\rm Area} \Sigma_{\mc C}} \| F_A \|_{L^2(\Sigma_{\mc C})} \leq  |\xi_W| \sqrt{{\rm Area} \Sigma_{\mc C}} \sqrt{E}.
 \eeqn
Since ${\rm Area} \Sigma_{\mc C}$ has a lower bound, the number $C$ cannot exceed a constant depending on $E$. Therefore, whether ${\mc C}$ has punctures or not, $|{\mc F}_W(u)|$ has an upper bound. By the properness of ${\mc F}_W|_{\ov{({\rm Crit} W)^{\rm ss}}}$, we have proved Theorem \ref{thm43}.
\end{proof}

\subsection{Energy bound}

We prove that solutions to the gauged Witten equation has an {\it a priori} bound on their energy.

\begin{thm}\label{energybound}
For each equivariant curve class $\ubar B \in H_2^{\ubar K}(V; {\mb Z})$ there is a positive constant $E(\ubar B) > 0$ satisfying the following condition. Suppose ${\bm v} = (P, A, u)$ is a bounded solution to the gauged Witten equation over a smooth $r$-spin curve ${\mc C}$ which represents the class $\ubar B$ and whose image is contained in $\ubar P(\ov{({\rm Crit} W)^{\rm ss}})$. Then 
\beqn
E(P, A, u) \leq E(\ubar B).
\eeqn
\end{thm}

\begin{proof}
Choose a local coordinate chart $U\subset \Sigma_{\mc C}^*$ with coordinates $z  = s + \i t$. Choose a local trivialization of $P$, so that $u$ is identified as a map $u(s, t)$ and $\ubar A$ is identified with $d + \ubar \phi ds + \ubar \psi dt$. There is a closed 2-form on $U \times V$ which transforms naturally with change of local trivializations. Define
\beqn
\omega_{\ubar A} = \omega_V - d \langle \ubar \mu ,  \ubar \phi ds + \ubar \psi dt \rangle \in \Omega^2(U\times V)
\eeqn
where $\ubar \mu = (\mu_R, \mu): V \to \ubar {\mf k} \cong {\bf i} {\mb R} \oplus {\mf k}$ is the moment map for the $\ubar K$-action. It is easy to see that $\omega_{\ubar A}$ is a well-defined closed form on the total space $\ubar P(V)$. Suppose locally the volume form is $\sigma ds dt$. Since $u$ is holomorphic, one has
\beqn
\begin{split}
\frac{1}{2} |d_A u|^2  = &\ \frac{1}{2\sigma} \Big[ |\partial_s u + {\mc X}_{\ubar \phi} (u)|^2 + |\partial_t u + {\mc X}_{\ubar \psi} (u)|^2 \Big]\\
                       = &\ \frac{1}{\sigma}  \omega_V \Big( \partial_s u + {\mc X}_{\ubar \phi} (u), J (\partial_s u + {\mc X}_{\ubar \phi} (u)) \Big)\\                       = &\ * \Big[ u^* \omega_V - d (\mu \cdot ( \ubar \phi ds + \ubar \psi dt)) +  \ubar \mu (u) \cdot (\partial_s \ubar \psi - \partial_t \ubar \phi + [\ubar \phi, \ubar \psi ] ) \Big]\\
											= &\ * \Big[ u^* \omega_{\ubar A} + \ubar \mu (u) \cdot F_{\ubar A} \Big].
\end{split}
\eeqn
Therefore, by definition, the holomorphicity of $u$, and the vortex equation, one has
\beq\label{eqn42}
\begin{split}
E(P, A, u)  = &\ \frac{1}{2} \Big[ \| d_A u\|_{L^2}^2 + \| F_A \|_{L^2}^2 + \| \mu (u)\|_{L^2}^2 \Big]\\
         = &\ \int_{\Sigma_{\mc C}^*} u^* \omega_{\ubar A} + \int_{\Sigma_{\mc C}^*}  \ubar \mu (u) \cdot F_{\ubar A} + \frac{1}{2} \| F_A  \|_{L^2}^2 + \frac{1}{2} \|\mu (u)\|_{L^2}^2 \\
         = &\ \int_{\Sigma_{\mc C}^*} u^* \omega_{\ubar A} + \int_{\Sigma_{\mc C}^*} \mu_R (u) \cdot F_{A_{\mc C}} + \frac{1}{2} \| * F_A + \mu(u) \|_{L^2}^2 \\
         = &\ \int_{\Sigma_{\mc C}^*} u^* \omega_{\ubar A} +  \int_{\Sigma_{\mc C}^*} \mu_R (u) \cdot F_{A_{\mc C}}.
         \end{split}
         \eeq
It is a well-known fact that the integral of $u^* \omega_{\ubar A}$ over $\Sigma_{\mc C}$ is equal to the topological pairing $\langle [\omega_V^{\ubar K}], \ubar B \rangle$. On the other hand, by Lemma \ref{lemma211} and Theorem \ref{thm43}, the last term of \eqref{eqn42} is bounded by a constant $C>0$. Hence this theorem holds for $E(\ubar B) = \langle [\omega_V^{\ubar K}, \ubar B \rangle + C$.  
\end{proof}

When the domain has no markings, it is not obvious to us how to exclude solutions that are contained in the unstable locus of the critical point set. However we show that when the domain area is sufficiently large, solutions which are contained in $\ov{({\rm Crit} W)^{\rm ss}}$ are separated from solutions contained in the unstable locus.

\begin{cor}
Given $\ubar B \in H_2^{\ubar K}(V; {\mb Z})$ and $g \geq 2$ there exists $C(\ubar B)>0$ satisfying the following condition. Suppose ${\mc C}$ is an $r$-spin curve with no markings and $(P, A,u)$ is a solution to \eqref{eqn39}. Suppose $u( \Sigma_{\mc C}) \subset \ubar P ( \ov{ ({\rm Crit} W)^{\rm ss}})$ and ${\rm Area}\Sigma_{\mc C} \geq C(\ubar B)$, then 
\beqn
u(\Sigma_{\mc C}) \cap \ubar P ( ({\rm Crit} W)^{\rm ss}) \neq \emptyset.
\eeqn
\end{cor}

\begin{proof}
Assume $(P, A, u)$ satisfies the hypothesis. Theorem \ref{energybound} shows that 
\beqn
E(P, A, u) \leq E(\ubar B).
\eeqn
By definition, one has
\beqn
E(\ubar B) \geq E(P, A, u) \geq \int_{{\mc C}} |\mu(u)|^2.
\eeqn
Then if the area of ${\mc C}$ is sufficiently large, the minimal value of $|\mu(u)|$ can be arbitrarily small. Therefore, the image of $u$ must intersect a neighborhood of $\mu^{-1}(0)$. The corollary then follows. 
\end{proof}

\begin{rem}\label{rem46}
When the number of markings $n$ is zero, we will assume that the minimal area of a genus $g$ curve with respect to our chosen metric is sufficiently large. From now on we will always assume that solutions to the gauged Witten equation are contained in $\ov{({\rm Crit} W)^{\rm ss}}$. Then using the morphism $\iota_W: {\mb C}^* \to G$ assumed in Hypothesis \ref{hyp34}, one can convert a solution $(P, A, u)$ to the gauged Witten equation to a solution $(P', A', u')$ to the equation 
\beq\label{eqn43}
\left\{ \begin{array}{rcl}
\ov\partial_{A'} u'  &= & 0,\\
* F_{A'} + \mu(u') & = & * F_{\iota_W(A_{\mc C})}.  
\end{array} \right.
\eeq
This equation can be viewed as a perturbation of the symplectic vortex equation \eqref{eqn32} for the gauge group $K$ and the perturbation term $* F_{\iota_W(A_{\mc C})}$ decays exponentially over the cylindrical ends. Therefore results about vortices can be applied to the gauged Witten equation with very little modification.
\end{rem}

\section{A Compactification of the Moduli Space}\label{section5}

In the last section we considered solutions of the gauged Witten equation over a smooth $r$-spin curve. We have see that solutions are contained in the critical locus of $W$ and are holomorphic. Hence in particular they are all solutions to the symplectic vortex equation with target $V$. So the topology on their moduli space can be defined as the same as if it is a subset of the moduli space of vortices.

\subsection{Solitons}

Since we assume that $V$ is aspherical, sphere bubbling cannot happen. However, solutions can still bubble off solitons at punctures or nodes. The solitons are bounded solutions to the gauged Witten equation over the infinite cylinder. Notice that over the infinite cylinder, the log-canonical bundle is trivialized by $ds + {\bf i} dt$ where $s + {\bf i} t$ is the standard cylindrical coordinate. Then the set of isomorphism classes of $r$-spin structures over $\Thetait$ is a group isomorphic to ${\mb Z}_r$. By our convention (see Subsection \ref{subsection21}) the line bundle $L_{\mc C}$ which is trivial has a flat Hermitian metric with Chern connection written as 
\beqn
A_{\mc C} = d + {\bf i} \frac{m}{r} dt,\ m \in \{ 0, 1, \ldots, r-1\}.
\eeqn

\begin{defn}\label{defn51}
An $m$-{\it soliton} is a bounded solution to gauged Witten equation over the infinite cylinder $\Thetait = {\mb R} \times S^1$ equipped with translation invariant metric. More precisely, it is a triple ${\bm v} = (u, \phi, \psi): \Thetait \to {\rm Crit} W \times {\mf k} \times {\mf k}$ solving the equation
\beqn
\left\{ \begin{array}{rcc} \partial_s u + {\mc X}_\phi(u) + J (\partial_t u + {\mc X}_{\frac{\i m}{r} + \psi}(u)) & = & 0,\\
                          \partial_s \psi - \partial_t \phi + [\phi, \psi] + \mu(u) &  = & 0 \end{array} \right.
\eeqn
such that its energy is finite and its image has a compact closure. 
\end{defn}

By the basic result about the asymptotic behavior of vortices (see for example \cite{Venugopalan_quasi}\cite{Chen_Wang_Wang}), one can gauge transform a soliton so that
\begin{enumerate}
\item as $s\to \pm\infty$, $\phi \to 0$ and $\psi$ converges to $\pm \lambda_m \in {\mf k}$. 

\item There exists $x_\pm \in {\rm Crit} W \cap \mu^{-1}(0)$ such that
\beqn
u(s, t) = x_\pm.
\eeqn
\end{enumerate}
So $u$ extends to a continuous section in the orbifold sense. Therefore, a soliton represents an equivariant curve class $\ubar B \in H_2^{\ubar K} ( V; {\mb Z})$ and we have the following energy identity. 

\begin{prop}\label{prop52} {\rm (Energy identity for solitons)}
Given a soliton ${\bm v} = (u, \phi, \psi)$ that represents an equivariant curve class $\ubar B \in H_2^{\ubar K} ( V; {\mb Z})$, one has
\beqn
E(u, \phi, \psi) = \langle [\omega_V^{\ubar K}], \ubar B \rangle.
\eeqn
\end{prop}
\begin{proof}
Using the same calculation of the proof of Theorem \ref{energybound}, we see in \eqref{eqn42} the curvature term $F_{A_{\mc C}}$ vanishes. Hence we obtain an equality. 
\end{proof}

Moreover, for the purpose of proving compactness, one also need to prove that the energy of a nontrivial soliton is bounded from below. 

\begin{thm}\label{thm53}
There exists $\epsilon_W>0$ that only depends on the GLSM space, such that for any soliton ${\bm v}$ with positive energy, one has $E({\bm v}) \geq \epsilon_W$.
\end{thm}

\begin{proof}
One can use a standard argument to prove an estimate such as
\beqn
E({\bm v}; [s_0 + T, s_1 - T] \times S^1) \leq C e^{- \delta T} E({\bm v}; [s_0, s_1] \times S^1)
\eeqn
for constants $C, \delta>0$ independent of ${\bm v}$, whenever the total energy of ${\bm v}$ is small enough. It then implies that $E({\bm v})$ has a positive lower bound. The details are left to the reader.
\end{proof}

\subsection{Decorated dual graphs}

We first fix notations about graphs. In this paper, a graph $\Gamma$ consists of a set of vertices $\V(\Gamma)$, a set of edges $\E(\Gamma)$, and a set of tails $\T(\Gamma)$ with several structural maps describing how each edge connects two vertices and which vertex a tail is attached to. The set of tails $\T(\Gamma)$ are always identified with $\{1, 2, \ldots, n\}$ for some $n \geq 0$ and this identification is implicitly regarded as part of the data. For each graph $\Gamma$, there is an induced graph $\tilde \Gamma$ obtained by cutting off all edges in $\E(\Gamma)$ so that each edge becomes two new tails of $\tilde \Gamma$. The operation $\Gamma \to \tilde \Gamma$ corresponds to the normalization of nodal curves. Denote 
\beqn
\tilde \E(\Gamma):= \T(\tilde \Gamma) \setminus \T(\Gamma).
\eeqn
This is the set of tails obtained from cutting edges of $\Gamma$ so $\# \tilde \E(\Gamma) = 2 \# \E(\Gamma)$. The {\it valence} of a vertex $v \in \V(\Gamma)$, denoted by 
\beqn
\tt{d}(v) \in \{0, 1, 2, \cdots\},
\eeqn
is defined to be the number of tails in $\tilde \Gamma$ that are attached to $v$. 

To describe the combinatorial types of $r$-spin curves we need graphs together with extra structures. 

\begin{defn}
An $r$-spin dual graph (a dual graph for short) is a tuple 
\beqn
(\Gamma, \tt{g}, \tt{m})
\eeqn
where 
\begin{align*}
&\ \tt{g}: \V(\Gamma) \to \{ 0, 1, 2, \cdots,\},\ &\ \tt{m}: \T(\tilde \Gamma) \to {\mb Z}_r 
\end{align*}
are maps. The following conditions are required.
\begin{enumerate}
\item If $\tilde e_\pm \in \tilde \E (\Gamma)$ are obtained by cutting an edge, then $\tt{m}(\tilde e_-) \tt{m}(\tilde e_+) = 1$.

\item If $\tt{g}(v) = 0$, then $\tt{d}(v) \geq 2$; if $\tt{g}(v) = 1$, then $\tt{d}(v) \geq 1$.

\item When $\tt{g}(v) = 0$ and $\tt{d}(v) = 2$, let $t, t'\in T(\tilde \Gamma)$ be the two tails in the normalization which are attached to $v$. Then $t$ and $t'$ are not the same edge in $\Gamma$; moreover, $\tt{m}(t) \tt{m}(t') = 1$.
\end{enumerate}

A vertex $v \in \V(\Gamma)$ is called {\it cylindrical} if $\tt{g}(v) = 0$ and $\tt{d} (v) = 2$. An $r$-spin dual graph is {\it stable} if there is no cylindrical vertices.
\end{defn}

Now we define the notion of maps between dual graphs. 

\begin{defn}\label{defn55}
A map from an $r$-spin dual graph $\Pi$ to another $r$-spin dual graph $\Gamma$ consists of a surjective map $\rho_\V: \V(\Pi) \to \V(\Gamma)$, an injective map $\rho_{\tilde \E}: \tilde \E(\Gamma) \to \tilde \E (\Pi)$, and a bijective map $\rho_\T: \T(\Gamma) \to \T(\Pi)$ satisfying the following properties. 
\begin{enumerate}

\item If $t \in \T(\Gamma)$ is attached to $v\in \V(\Gamma)$, then $\rho_\T(t)$ is attached to $\rho_\V(v)$. 

\item If $\tilde e_-, \tilde e_+\in \tilde \E(\Gamma)$ are obtained from cutting an edge, so are $\rho_{\tilde \E}(\tilde e_-)$ and $\rho_{\tilde \E}(\tilde e_+)$. In other words, $\rho_{\tilde \E}$ is a lift of an injective map $\rho_\E: \E(\Gamma) \to \E(\Pi)$. 

\item If $v, v' \in \V(\Pi)$ are adjacent, then either $\rho_\V(v) = \rho_\V(v')$ or they are adjacent in $\Gamma$; moreover, if $v, v'$ are connected by $\rho_\E(e)$, then $\rho_\V(v)$ and $\rho_\V(v')$ are connected by $e$.

\item If $e \in \E(\Pi)$ connects $v, v' \in \V(\Pi)$ and is not in the image of $\rho_\E$, then $\rho_\V(v) = \rho_\V(v')$. 

\item For each $v\in \V(\Gamma)$, let $\Gamma_v$ be the maximal subgraph of $\Gamma$ which has a single vertex $v$ and let $\Pi_{\rho_\V^{-1}(v)}$ be the maximal subgraph of $\Pi$ whose vertices are $\rho_\V^{-1}(v)\subset \V(\Pi)$. Then 
\beqn
\tt{g}(\Gamma_v) = \tt{g}(\Pi_{\rho_\V^{-1}(v)}).
\eeqn

\item The maps $\rho_\T$ and $\rho_{\tilde \E}$ preserve the map $\tt{m}$.
\end{enumerate}
Here ends Definition \ref{defn55}.
\end{defn}

A partial order $\preq$ among $r$-spin dual graphs is defined as follows: $\Gamma \preq \Pi$ if there is a map from $\Gamma$ to $\Pi$. 

One can use stable $r$-spin dual graphs to label strata of the moduli space $\ov{\mc M}{}_{g,n}^r$. For an $r$-spin dual graph $\Gamma$, denote by $|\Gamma|$ the 1-complex associated to $\Gamma$. The genus of $\Gamma$ is defined as 
\beqn
\tt{g}(\Gamma):= \sum_{v\in \V(\Gamma)} \tt{g}(v) + {\rm rank} H_1( |\Gamma|).
\eeqn
If $\Gamma$ has genus $g$, $n$ tails, and is stable, then one can use $\Gamma$ to label a stratum ${\mc M}_\Gamma^r \subset \ov{\mc M}{}_{g, n}^r$. By abuse of notation, let $\Gamma$ also denote the datum by forgetting the map $\tt{m}$, which can be used to label a stratum ${\mc M}_\Gamma \subset \ov{\mc M}_{g, n}$ of the Deligne--Mumford space. By using the partial order $\preq$, the closure of ${\mc M}{}_{\Gamma}^r$ can be described as 
\beqn
\ov{\mc M}{}_{\Gamma}^r = \bigsqcup_{\Pi \preq \Gamma, \Pi = \Pi^{\rm st} } {\mc M}{}_{\Pi}^r. 
\eeqn

Moreover, consider a stable $r$-spin curve ${\mc C}$ of type $\Gamma$. The space of gluing parameters ${\mc V}_{\rm res}$ is stratified in an obvious way: a stratum is the subspace defined by the vanishing of gluing parameters at a subset of nodes. We use letters $\alpha, \beta, \cdots$ to label strata of ${\mc V}_{\rm res}$. Then each stratum $\alpha$ describes a particular way of resolving nodes of ${\mc C}$, hence induces an $r$-spin dual graph $\Gamma^\alpha$ with a morphism
\beq\label{eqn51}
\rho^\alpha: \Gamma \to \Gamma^\alpha,\ \ \alpha \in {\rm Strata} ({\mc V}_{\rm res}).
\eeq
If $\eta \in {\mc V}_{\rm def}$ and $\zeta \in {\mc V}_{\rm res}^\alpha\subset {\mc V}_{\rm res}$, then one has
\beqn
[{\mc C}_{\eta, \zeta}] \in {\mc M}_{\Gamma^\alpha}^r.
\eeqn

Given an $r$-spin dual graph $\Gamma$, there is a notion of {\it stabilization}, denoted by $\Gamma^{\rm st}$. Notice that there is no morphism $\Gamma \to \Gamma^{\rm st}$ in any sense. 

\begin{lemma}\label{lemma56}
Let $\rho: \Gamma \to \Pi$ be a map between $r$-spin dual graphs. Then it induces a canonical map $\rho^{\rm st}: \Gamma^{\rm st} \to \Pi^{\rm st}$ between their stabilizations. 
\end{lemma}

\begin{proof}
Left to the reader.
\end{proof}

The combinatorial type of such smooth or nodal $r$-spin curves can be described by {\it decorated dual graphs}. 

\begin{defn}\label{defn57} {\rm (Decorated dual graph)}
A {\it decorated dual graph} is a tuple 
\beqn
{\sf \Gamma} = \Big( \Gamma, (\ubar B{}_v)_{v \in \V(\Gamma)} \Big)
\eeqn
where $\Gamma$ is an $r$-spin dual graph and $\ubar B{}_v\in H_2^{\ubar K}( V; {\mb Z})$ is a collection of equivariant curve classes indexed by all vertices $v \in \V(\Gamma)$. A decorated dual graph is {\it stable} if for each cylindrical vertex $v$, $\ubar B{}_v \neq 0$. 
\end{defn}

There is a natural map ${\sf \Gamma} \mapsto \Gamma$ by forgetting the decorations. The partial order can be lifted to a partial order among all decorated dual graphs, which is still denoted by $\sf\Pi \preq \sf\Gamma$. We omit the detail of the definition.

\subsection{Stable solutions}\label{subsection53}

\begin{defn}\label{defn58} {\rm (Stable solutions)}
Let $\sf \Gamma$ be a decorated dual graph. A {\it solution to the gauged Witten equation of combinatorial type $\sf\Gamma$} consists of a smooth or nodal $r$-spin curve ${\mc C}$ of combinatorial type $\Gamma$, and a collection of objects
\beqn
{\bm v}:= \Big[ ({\bm v}_v)_{ v \in \V(\Gamma)} \Big]
\eeqn
Here for each stable vertex $v \in \V (\Gamma)$, ${\bm v}_v = (P_v, A_v, u_v)$ is a solution to the gauged Witten equation over the smooth $r$-spin curve ${\mc C}_v$; for each cylindrical vertex $v \in \V (\Gamma)$, ${\bm v}_v = (u_v, \phi_v, \psi_v)$ is a soliton. They satisfy the following conditions.
\begin{enumerate}
\item For each vertex $v \in \V (\Gamma)$, the equivariant curve class represented by ${\bm v}_v$ coincides with $\ubar B{}_v$ (which is contained in the data $\sf \Gamma$).

\item For each edge $e \in \E(\Gamma)$ corresponding to a node in ${\mc C}$, let $\tilde e_-$ and $\tilde e_+$ be the two tails of $\tilde \Gamma$ obtained by cutting $e$, which are attached to vertices $v_-$ and $v_+$ (which could be equal), corresponding to preimages $\tilde w_-$ and $\tilde w_+$ of $w$ in $\tilde {\mc C}$. Then the evaluations of ${\bm v}_{v_-}$ at $w_-$ and ${\bm v}_{v_+}$ at $w_+$ (which exist by Theorem \ref{thm42}) are equal.

\end{enumerate}

The solution is called {\it stable} if $\sf\Gamma$ is stable. The energy (resp. curve class) of a solution ${\bm v}$ is the sum of energies (resp. curve classes) of each component.
\end{defn}

One can define an equivalence relation among all stable solutions, and define the corresponding moduli spaces of equivalence classes. 

\begin{defn}\label{defn59}
Given a decorated dual graph $\sf\Gamma$, let 
\beqn
{\mc M}_{\sf\Gamma}:= {\mc M}{}_{\sf\Gamma} (V, G, W, \mu ) 
\eeqn
be the set of equivalence classes of stable solutions of combinatorial type $\sf\Gamma$. Denote
\beqn
\ov{\mc M}{}_{g, n}^r(V, G, W, \mu) = \bigsqcup_{\sf\Gamma} {\mc M}_{\sf\Gamma}
\eeqn
where the disjoint union is taken for all decorated dual graphs with genus $g$ and $n$ tails. For each degree $\ubar B \in H_2^{\ubar K}(V; {\mb Z})$, denote by
\beqn
\ov{\mc M}{}_{g,n}^r(V, G, W, \mu; \ubar B) \subset \ov{\mc M}{}_{g,n}^r(V, G, W, \mu)
\eeqn
the subset corresponding to types with curve class $\ubar B$. For each $E>0$, denote by 
\beqn
\ov{\mc M}{}_{g,n}^r(V, G, W, \mu)_{\leq E} \subset \ov{\mc M}{}_{g,n}^r(V, G, W, \mu)
\eeqn
the subset of elements represented by solutions whose energy is at most $E$.
\end{defn}

\subsection{Compactness of the moduli spaces}

As usual, the topology of the moduli spaces $\ov{\mc M}{}_{g, n}^r(V, G, W, \mu)$ is induced from a notion of sequential convergence. The reader can compare with the notion of convergence defined by Venugopalan \cite[Definition 3.4]{Venugopalan_quasi}.

We first define the convergence for a sequence of solutions defined on smooth stable $r$-spin curves.

\begin{defn} {\rm (Convergence of smooth solutions)} \label{defn510}
Let ${\mc C}_k$ be a sequence of stable smooth $r$-spin curves of genus $g$ with $n$ markings. Let ${\bm v}_k$ be a sequence of stable solutions over ${\mc C}_k$. Let ${\mc C}$ be a smooth or nodal $r$-spin curve of genus $g$ with $n$ markings and ${\bm v}$ be a stable solution over ${\mc C}$. We say that ${\bm v}_k$ {\it converges} to ${\bm v}$ if the following conditions hold. 
\begin{enumerate}

\item The isomorphism classes of ${\mc C}_k$ in $\ov{\mc M}{}_{g, n}^r$ converges to that of ${\mc C}^{\rm st}$. By Lemma \ref{domainconverge}, there exist a sequence of deformation parameters $\eta_k$, a sequence of gluing parameters $\zeta_k$, a sequence of isomorphisms 
\beqn
{\mc C}_k \cong {\mc C}_{\eta_k, \zeta_k}.
\eeqn

\item Let the dual graph of ${\mc C}$ be $\Gamma$. For each stable vertex $v \in \V(\Gamma^{\rm st})$, for any compact subset $Z \subset \Sigma_{{\mc C}, v}^*$ which is disjoint from markings and nodes, there is a bundle isomorphism $g_k: P|_Z \to P_k|_{Z_k}$ that covers $\iota_k$ such that $\ubar g_k^* {\bm v}_k|_Z$ converges to ${\bm v}|_Z$. Here $\ubar g_k: \ubar P|_Z \cong \ubar P_k|_Z$ is induced from $g_k$ and the isomorphism $P_{\mc C}|_Z \cong P_{{\mc C}_k}|_Z$. Here we identified $Z$ with a subset $Z_k:= Z_{\eta_k, \zeta_k} \subset \Sigma_{{\mc C}_{\eta_k, \zeta_k}}^*$ using the convention of Notation \ref{notation27}. 

\item For each cylindrical vertex $v\in \V(\Gamma)$, it is mapped under the stabilization map ${\mc C} \to {\mc C}^{\rm st}$ to either a marking or a node. Then there exists a sequence of points $y_k  = s_k + {\bf i} t_k$ in the cylinder satisfying the following condition. For any $R>0$, the cylinder $[s_k-R, s_k + R]\times S^1$ is contained in the long cylinder $(-T_{\eta_k, \zeta_k}(w) , T_{\eta_k, \zeta_k}(w)) \times S^1$ (see Notation \ref{notation27}). The coordinates induce trivializations of $P_{{\mc C}_k}$ over the cylinder. Moreover, there is a sequence of bundle isomorphisms $g_k: P_v|_{[-R, R]\times S^1} \to P_k|_{[s_k - R, s_k + R] \times S^1}$ (where $P_v$ is trivialized) that covers the map $y \mapsto y + y_k$ such that $\ubar g_k^* {\bm v}_k$ converges in compact convergence topology, meaning that over any precompact open subset of the domain, $\ubar g_k^* {\bm v}_k$ converges uniformly with all derivatives, to ${\bm v}_v |_{[-R, R] \times S^1}$. Here $\ubar g_k$ is the bundle isomorphisms induced from $g_k$ and the trivializations of $P_{{\mc C}_k}$. 

\item There is no energy lost, namely
\beqn
E({\bm v}) = \lim_{k \to \infty} E({\bm v}_k).
\eeqn
\end{enumerate}
Here ends Definition \ref{defn510}.
\end{defn}

We also define the sequence convergence of solitons. First, a {\it broken soliton} is a finite sequence of solitons ${\bm v}_1 \# \cdots \# {\bm v}_s$ satisfying the following conditions. 
\begin{enumerate}
\item The $r$-spin structures of each ${\bm v}_i$ are labelled by the same $m \in {\mb Z}_r$. 

\item Each ${\bm v}_i$ is a nontrivial soliton, i.e., having positive energy.

\end{enumerate}

Now we define the convergence of solitons (see Definition \ref{defn51} for the definition of solitons). 

\begin{defn} \label{defn511} {\rm (Convergence of solitons)}
Let ${\bm v}_k$ be a sequence of $q$-solitons with $q = \frac{m}{r}$ for some $m \in \{0, 1, \ldots, r-1\}$. We say that ${\bm v}_k$ converges to a broken $q$-solitons ${\bm v}_1 \# \cdots \# {\bm v}_s$ if the following conditions hold. 
\begin{enumerate}

\item For $b = 1, \ldots, s$, there exist a sequence of points $y_{b, k} = s_{b, k} + {\bf i} t_{b, k} \in {\mc C}_k$ and a sequence of gauge transformations $g_k: \Thetait \to K$ such that $g_{b, k}^* \phi_{b, k}^* {\bm v}_k$ converges in compact convergence topology. to ${\bm v}_b$. Here $\phi_{b, k}(y) = y + y_{b, k}$ is the translation on the cylinder. 

\item The sequences of points $y_{b, k}$ satisfy 
\beqn
b < b' \Longrightarrow \lim_{k \to \infty} s_{b', k} - s_{b, k} = +\infty.
\eeqn

\item There is no energy lost, namely
\beqn
E({\bm v}) = \lim_{k \to \infty} E({\bm v}_k).
\eeqn
\end{enumerate}
Here ends Definition \ref{defn511}.
\end{defn}

Now we define sequential convergence of general stable solutions.

\begin{defn}\label{defn512}{\rm (Convergence of stable solutions)}
Let ${\mc C}_k$ be a sequence of smooth or nodal $r$-spin curve of type $(g, n)$ and ${\bm v}_k$ be a sequence of stable solutions over ${\mc C}_k$. Let ${\mc C}$ be another smooth or nodal $r$-spin curve of type $(g, n)$ and ${\bm v}$ be a stable solution over ${\mc C}$. We say that $({\mc C}_k, {\bm v}_k)$ converges to $({\mc C}, {\bm v})$ if after removing finitely many elements in the sequence, the sequence can be divided into finitely many subsequences, and each subsequence satisfies the following conditions. Without loss of generality, we assume there is only one subsequence, which is still indexed by $k$.
\begin{enumerate}

\item The underlying $r$-spin dual graphs of ${\mc C}_k$ are all isomorphic. Denote it by $\Gamma$. Let the combinatorial type of ${\mc C}$ be $\Pi$.

\item The sequence of stabilizations ${\mc C}_k^{\rm st}$ converge to the stabilization ${\mc C}^{\rm st}$. Therefore there exist a sequence of deformation parameters $\eta_k$, a sequence of gluing parameters $\eta_k$, and a sequence of isomorphisms of $r$-spin curve ${\mc C}_k^{\rm st} \cong {\mc C}^{\rm st}_{\eta_k, \zeta_k}$. Moreover, all $\zeta_k$ belong to the same stratum $\alpha$ of the space of gluing parameters of ${\mc C}$ hence the isomorphisms ${\mc C}_k^{\rm st} \cong {\mc C}^{\rm st}_{\eta_k, \zeta_k}$ induce a fixed morphism $\rho^{\rm st}: \Pi^{\rm st} \to \Gamma^{\rm st}$ (see Definition \ref{defn55} and \eqref{eqn51}). 

\item There is a map of $r$-spin dual graphs $\rho: \Gamma \to \Pi$ such that its induced map between stabilizations (see Lemma \ref{lemma56}) is $\rho^{\rm st}$.

\item For each vertex $v \in \V(\Gamma)$, let the corresponding irreducible component in ${\mc C}_k$ be ${\mc C}_{k, v}$ and let $\tilde {\bm v}_{k, v}$ be the restriction of ${\bm v}_k$ to ${\mc C}_{k, v}$. Let $\Pi_v$ be the maximal subgraph of $\Pi$ whose set of vertices is $\rho_\V^{-1}(v)$ and let $({\mc C}_{\Pi_v}, {\bm v}_{\Pi_v})$ be the restriction of $({\mc C}, {\bm v})$ to $\Pi_v$. We require that ($v$ is either cylindrical or not), the sequence $({\mc C}_{k, v}, {\bm v}_{k, v})$ converges to $({\mc C}_{\Pi_v}, {\bm v}_{\Pi_v})$ in the sense of Definition \ref{defn510} or Definition \ref{defn511}. 

\end{enumerate}
\end{defn}

One can also show that the notion of sequential convergence defined by Definition \ref{defn512} descends to a notion of sequential convergence in the moduli space $\ov{\mc M}{}_{g, n}^r(V, G, W, \mu)$. We explain why the notion of sequential convergence induces a unique topology on the moduli spaces. The reason is similar to the case of Gromov--Witten theory and it has been explained in \cite[Section 5.6]{McDuff_Salamon_2004} and \cite{McDuff_Salamon_erratum_2} in the context for pseudoholomorphic spheres. In our case, since we can regard solutions to the gauged Witten equation as vortices with structure group $\ubar K$, the topology is induced in the same way as \cite[Section 5]{Venugopalan_quasi}. We state the result as the following proposition.

\begin{prop}
Define a topology on $\ov{\mc M}{}_{g, n}^r(V, G, W, \mu)$ as follows. A subset of $\ov{\mc M}{}_{g, n}^r(V, G, W, \mu)$ is closed if it is closed under the sequential convergence. Then this topology is Hausdorff, first countable, and its set of converging sequences coincides with the set of converging sequences defined by Definition \ref{defn512}.
\end{prop}	

The next task is to prove the (sequential) compactness. 

\begin{thm} \label{thm514} {\rm (Compactness)}
\begin{enumerate}

\item For any $E>0$, $\ov{\mc M}{}_{g,n}^r(V, G, W, \mu)_{\leq E}$ is compact. 

\item For any stable decorated dual graph $\sf\Gamma$ with $n$ tails and genus $g$, $\ov{\mc M}{}_{\sf \Gamma}(V, G, W, \mu)$ is a compact subset of $\ov{\mc M}{}_{g, n}^r (V, G, W, \mu)$. 
\end{enumerate}
\end{thm}

\begin{proof}
One only needs to prove the first item since the second item follows from the first and Theorem \ref{energybound}. To prove the first item, one only needs to prove sequential compactness. Indeed our case is very similar to the case of \cite{Venugopalan_quasi} so we will be sketchy. Let $({\mc C}_k, {\bm v}_k)$ be a sequence of stable solutions which represent a sequence of points in $\ov{\mc M}{}_{g, n}^r(V, G, W, \mu)_{\leq E}$. By Theorem \ref{thm41}, the gauged maps ${\bm v}_k$ are holomorphic. Then one can view them as solutions to the (perturbed) symplectic vortex equation (see Remark \ref{rem46}) and one can use the compactness result of vortices. We remind the reader of two facts we will use implicitly in the argument. First, by Theorem \ref{thm43}, given the energy bound, solutions are all contained in a compact subset of $V$. Second, by the aspherical condition of $V$ (see Hypothesis \ref{hyp34}), one does not need to worry about bubbling and hence with respect to the cylindrical metrics, there holds
\beq\label{eqn52}
\sup_k \| d_{A_k} u_k \|_{L^\infty} < \infty. 
\eeq

Now we start to construct the limit. First, the stabilized domain ${\mc C}_k^{\rm st}$ has a subsequence (still indexed by $k$) converging to a limiting stable $r$-spin curve, denoted by ${\mc C}^{\rm st}$. Then one can identify as $r$-spin curves ${\mc C}_k^{\rm st}$ with ${\mc C}^{\rm st}_{\eta_k, \zeta_k}$ for certain deformation parameters $\eta_k$ and gluing parameter $\zeta_k$. Moreover, by taking a further subsequence, we may assume that all $\eta_k$ belong to the same stratum of the space of gluing parameters ${\mc V}_{\rm res}$. One can choose an exhausting sequence of precompact open subsets $Z_1 \subset Z_2 \subset \cdots \subset \Sigma_{{\mc C}^{\rm st}}^*$. One can identify each $Z_i$ with a subset $Z_{i,k} \subset \Sigma_{{\mc C}_k^{\rm st}}^*$  for $k$ sufficiently large. Then over $Z_i$, by the vortex equation $* F_{A_k} + \mu(u_k) = 0$,  the uniform $C^0$-bound of $u_k$, and the convergence of the metric (see Definition \ref{defn29}), one has
\beqn
\sup_k \| F_{A_k}\|_{L^\infty(Z_i)} < \infty.
\eeqn
Then by the Uhlenbeck compactness, up to gauge transformations on $Z_i\cong Z_{i,k}$, a subsequence of $A_k$ converges weakly in $W^{1,p}$ (for certain $p>2$) to a limit $W^{1,p}$-connection $A$ on $Z_i$. On the other hand, the equation $\ov\partial_{A_k} u_k = 0$ can be viewed as a perturbed pseudoholomorphic curve equation where the perturbation term depends on the connections $A_k$ and $A_{{\mc C}_k}$. The weak convergence $A_k \to A$ and the convergence of $A_{{\mc C}_k}$ (see Lemma \ref{lemma212}) implies that the perturbation term converges at least in $C^0$. Moreover, by the uniform $C^0$-bound on $u_k$ and \eqref{eqn52}, one can construct a weak $W^{1,p}$-limit of $u_k$, denoted by $u$, over $Z_i$. The limit is smooth modulo gauge transformations; moreover, using the elliptic bootstrapping one can show that the convergence $(A_k, u_k)$ to $(A, u)$ is indeed $C^\infty$ over $Z_i$. One can repeat the argument for all $Z_i$ and select further subsequences and obtain the limiting solution over $\Sigma_{{\mc C}^{\rm st}}$. 

It remains to construct all soliton components of the limit. In ${\mc C}_k$ there are regions which are the infinite cylinder $(-\infty, +\infty) \times S^1$ corresponding to soliton components of ${\bm v}_k$, semi-infinite cylinders $(0, +\infty) \times S^1$ corresponding to punctures or {\it old} nodes of ${\mc C}^{\rm st}$ (the nodes at which the gluing parameters $\zeta_k$ are always zero), or finite cylinders $(-T_k, T_k) \times S^1$ with $T_k \to + \infty$ corresponding to {\it new} nodes of ${\mc C}^{\rm st}$ (the nodes at which the gluing parameters $\zeta_k$ are always nonzero). It is standard (see for example the proof of compactness in \cite{Xu_VHF} or \cite{Venugopalan_quasi}) to show that the solutions ${\bm v}_k$ restricted to these cylindrical regions converge (subsequentially) up to adding broken solitons. (One needs the convergence of the cylindrical metric (see Definition \ref{defn29}) and the convergence of the connection $A_{{\mc C}_k}$ (see Lemma \ref{lemma212}.) Theorem \ref{thm53} implies that the length of the limiting broken solitons are finite and no energy lost. 

This finishes the construction of the limit as a solution to the vortex equation. It is easy to recover from the $r$-spin structure of the limiting domain curve the limit as a solution to the gauged Witten equation and routine to verify that the object we constructed is indeed a limit in the sense of Definition \ref{defn512}. This finishes the proof.
\end{proof}

\section{Virtual Cycles and Cohomological Field Theory}\label{section6}

In this section we define the GLSM correlation functions under the assumption of the existence and properties of the virtual cycles on the moduli spaces. The virtual cycle construction and the proof of the properties are deferred to the companion paper \cite{Tian_Xu_geometric_2}. Lastly, we give a brief argument showing the relation between the GLSM correlation functions and the Gromov--Witten invariants of the classical vacuum and explain the connection with mirror symmetry.

\subsection{The virtual cycle and the GLSM correlation functions}

Let ${\sf \Gamma}$ be a stable decorated dual graph with $n$ tails (not necessarily connected). We have defined the compactification of the moduli space of gauge equivalence classes of solutions to the gauged Witten equation with combinatorial type $\sf\Gamma$, denoted by 
\beqn
\ov{\mc M}_{\sf \Gamma}:= \ov{\mc M}_{\sf \Gamma}(V, G, W, \mu). 
\eeqn
There is a natural {\it evaluation maps}
\beqn
{\rm ev}: \ov{\mc M}_{\sf\Gamma} \to X^n.
\eeqn
There is also the {\it forgetful map} (or called the {\it stabilization map})
\beqn
{\rm st}: \ov{\mc M}_{\sf \Gamma} \to \ov{\mc M}_{\Gamma} \subset \ov{\mc M}_{g, n}. 
\eeqn
Here $\Gamma$ is the underlying dual graph of $\sf\Gamma$, where unstable rational components are contracted, and ${\rm st}$ is defined by forgetting the $r$-spin structure, forgetting the gauged maps, and stabilization. Denote
\beqn
( \rm{ev}, {\rm st}): \ov{\mc M}_{\sf\Gamma} \to X^{\sf\Gamma} \times \ov{\mc M}_{\Gamma}.
\eeqn
In order to define the correlation functions, what we need is the pushforward of the virtual fundamental class of $\ov{\mc M}_{\sf\Gamma}$. The mathematical statement about the properties of such the pushforward is stated as follows. 

\begin{thm}\label{thm61}\cite{Tian_Xu_geometric_2}
For each stable decorated dual graph ${\sf \Gamma}$ with $n$ tails (not necessarily connected), there is a homology class 
\beqn
[\ov{\mc M}_{\sf \Gamma}]^{\rm vir} \in H_*( \ov{\mc M}_\Gamma; {\mb Q}) \otimes H_*(X^n; {\mb Q}).
\eeqn
The collection of these classes satisfy the following properties. 
\begin{enumerate}

\item \label{thm61a} {\bf (Dimension)} If $\sf\Gamma$ is connected, then dimension of the class $[\ov{\mc M}_{\sf \Gamma}]^{\rm vir}$
\beqn (2-2g) {\rm dim}_{\mc C} X + 2 {\rm deg} {\sf\Gamma} + {\rm dim}_{\mb R} \ov{\mc M}_\Gamma.
\eeqn
Here if the equivariant curve class of $\sf\Gamma$ is $\ubar B \in H_2^{\ubar K}(V; {\mb Z})$, then
\beqn
{\rm deg}{\sf\Gamma}:= \langle c_1^{\ubar K} (TV), \ubar B \rangle \in {\mb Z}.
\eeqn

\item \label{thm61b} {\bf (Disconnected Graph)} Let ${\sf \Gamma}_1, \ldots, {\sf \Gamma}_k$ be connected decorated dual graphs. Let ${\sf \Gamma} = {\sf \Gamma}_1 \sqcup \cdots \sqcup {\sf\Gamma}_k$ be the disjoint union. Then one has 
\beqn
[\ov{\mc M}_{\sf \Gamma}]^{\rm vir} = \prod_{\alpha=1}^k [\ov{\mc M}_{{\sf \Gamma}_\alpha}]^{\rm vir}.
\eeqn

\item \label{thm61c} {\bf (Cutting Edge)} Let $\sf \Gamma$ be the decorated dual graph with $n$ markings and let $\sf \Pi$ be the graph obtained from $\sf \Gamma$ by shrinking a loops. Then one has 
\beqn
( \iota_{\Pi} )_* [ \ov{\mc M}_{\sf \Pi} ]^{\rm vir} = [ \ov{\mc M}_{\sf \Gamma} ]^{\rm vir} \cap [ \ov{\mc M}_{ \Pi} ];
\eeqn
Here the map 
\beqn
(\iota_{ \Pi} )_*: H_* ( X^n ) \otimes H_* ( \ov{\mc M}_{\Pi} ) \to H_* ( X^n ) \otimes H_* ( \ov{\mc M}_{ \Gamma} )
\eeqn
is induced from the inclusion $\iota_{\Pi}: \ov{\mc M}_{\Pi} \hookrightarrow \ov{\mc M}_{ \Gamma}$. 

\item \label{thm61d} {\bf (Composition)} Let $\sf \Pi$ be a decorated dual graph with $n$ markings with one distinguished edge. Let $\tilde {\sf \Pi}$ be the decorated dual graph obtained from $\sf \Pi$ by replacing this distinguished edge by two tails. Let $\Delta_X$ be the diagonal of $X$. Notice that there is a natural isomorphism 
\beqn
\ov{\mc M}_{\Pi} \cong \ov{\mc M}_{{\tilde \Pi}}.
\eeqn
Then one has
\beqn
[ \ov{\mc M}_{\sf \Pi} ]^{\rm vir} = [ \ov{\mc M}_{\tilde {\sf \Pi}} ]^{\rm vir} \setminus {\rm PD} (\Delta_X ).
\eeqn
Here we use the slant product 
\beqn
\setminus: H_* ( X^{n+2} )  \otimes H^* ( X \times X ) \to H_* ( X^n ).
\eeqn
\end{enumerate}
\end{thm}

From this theorem one can define a cohomological field theory on the cohomology of $X$ as follows. Indeed it is completely parallel to the case of Gromov--Witten invariants. Recall that the coefficient field is the Novikov field of formal Laurent series 
\beqn
\Lambda:= \Big\{ \sum_{i=1}^\infty a_i T^{\lambda_i} \ |\ \ a_i \in {\mb Q},\ \lim_{i \to \infty} \lambda_i = +\infty \Big\}
\eeqn
where $T$ is a formal variable. For each $g, n$ with $2g + n \geq 3$, consider the set of decorated dual graphs with only one vertex and no edges of genus $g$ with $n$ tails. Such decorated dual graphs are determined by an equivariant curve class $\ubar B$. Denote the corresponding moduli space by 
\beqn
\ov{\mc M}{}_{g, n}^r(\ubar B).
\eeqn
Then for classes 
\begin{align*}
&\ \alpha_a \in H^* ( X; {\mb Q} ),\ a = 1, \ldots, n,\  &\ \beta \in H^* ( \ov{\mc M}_{g, n}; {\mb Q} ),
\end{align*}
define 
\beq\label{eqn61}
\langle \alpha_1 \otimes \cdots \otimes \alpha_n; \beta \rangle_{g, n}^{\rm GLSM}:= \sum_{\ubar B} \left\langle [ \ov{\mc M}{}_{g, n}^r(\ubar B ) ]^{\rm vir}, \alpha_1 \otimes \cdots \otimes \alpha_n \otimes \beta \right\rangle T^{\langle [\omega_V^{\ubar K}], \ubar B \rangle}.
\eeq
Here we implicitly used the K\"unneth formula for the cohomology of products.

Before verifying the splitting properties, one needs to show that the right hand side of \eqref{eqn61} is an element of $\Lambda$. This follows from the compactness (Theorem \ref{thm514}). Indeed, below any energy level $E$, if there are infinitely different classes $\ubar B{}_k$ with $E(\ubar B{}_k) \leq E$ (see Theorem \ref{energybound} for the notation $E(\ubar B)$) which contribute to the expression (i.e., the moduli spaces are nonempty), then there is a sequence of solutions representing ${\ubar B}_k$. Then Theorem \ref{thm514} implies that a subsequence converges modulo gauge transformation to a stable solution. However, the convergence implies that the homology classes $\ubar B{}_k$ should stabilize for large $k$. This is a contradiction. Hence \eqref{eqn61} is well-defined. Then the correlation can extends to classes with $\Lambda$-coefficients and hence defines a multilinear function on the cohomology. 

We now verify the splitting properties (see Theorem \ref{thm11}). The first splitting property follows from the {\bf (Cutting Edge)} and the {\bf (Composition)} properties of the virtual cycles. Indeed, let $\ubar B \in H_2^{\ubar K} (V; {\mb Z})$ be an equivariant curve class, $g\geq 1$, $n \geq 0$ with $2g + n > 3$. Let $\sf \Gamma$ be the decorated dual graph with a single vertex of genus $g$ labelled by $\ubar B$ and with $n$. Let $\sf \Pi$ be the decorated dual graph obtained from $\sf \Gamma$ by shrinking a non-separating loop on the surface. Let
\beqn
\alpha_1, \ldots, \alpha_n \in H^* ( X; {\mb Q} ),\ \ \beta \in H^* ( \ov{\mc M}_{g, n}; {\mb Q} )
\eeqn
be cohomology classes. Then the {\bf (Cutting Edge)} property implies that 
\beqn
\left\langle [ \ov{\mc M}_{\sf \Pi} ]^{\rm vir}, \alpha_1 \otimes \cdots \otimes \alpha_n \otimes \iota_{\Pi}^* \beta \right\rangle = \left\langle  [ \ov{\mc M}_{\sf \Gamma} ]^{\rm vir}, \alpha_1 \otimes \cdots \otimes \alpha_n \otimes \beta \cup \gamma_\Pi \right\rangle.
\eeqn
Here $\gamma_\Pi \in H^*(\ov{\mc M}_{g,n}; {\mb Q})$ is the Poincar\'e dual to the divisor $\ov{\mc M}_\Pi \subset \ov{\mc M}_{g,n}$. On the other hand, let $\sf \Xi$ be the decorated dual graph obtained from $\sf \Gamma$ by replacing the node with two markings. Then the {\bf (Composition)} property implies that 
\beqn
\left\langle [ \ov{\mc M}_{\sf \Pi} ]^{\rm vir}, \alpha_1 \otimes \cdots \otimes \alpha_n \otimes \iota_{\Pi}^* \beta \right\rangle = \left\langle [ \ov{\mc M}_{\sf \Xi} ]^{\rm vir}, \alpha_1 \otimes \cdots \otimes \alpha_n \otimes {\rm PD}( \Delta_X ) \otimes \iota_{\Pi}^* \beta \right\rangle.
\eeqn
Hence 
\begin{multline}\label{eqn62}
\left\langle [ \ov{\mc M}_{\sf \Gamma} ]^{\rm vir}, \alpha_1 \otimes \cdots \otimes \alpha_n \otimes \beta \cup \gamma_\Pi \right\rangle \\
= \left\langle [ \ov{\mc M}_{\sf \Xi} ]^{\rm vir}, \alpha_1 \otimes \cdots \otimes \alpha_n \otimes {\rm PD} (\Delta_X )  \otimes \iota_\Pi^* \beta \right\rangle.
\end{multline}
Summing over all equivariant curve classes $\ubar B$, we obtain the first splitting property for non-separating nodes.

For the second splitting property, besides the {\bf (Cutting Edge)} and the {\bf (Composition)}, one also needs the {\bf (Disconnected Graph)} property. Indeed, using the same notations as above, with the constrain on $g$ and $n$ replaced by $g \geq 0$ and $n\geq 0$ with $2g + n > 3$, and instead of shrinking a non-separating loop in $\sf \Gamma$, we shrink a separating one. Then the decorated dual graph $\sf \Xi$ obtained by normalizing $\sf \Pi$ is isomorphic to a disjoint union $\sf \Gamma_1 \sqcup \sf \Gamma_2$. Then by the same argument as above, one still obtains \eqref{eqn62}. Now using the {\bf (Disconnected Graph)} property, one has
\beqn
 \left[ \ov{\mc M}_{\sf \Xi} \right]^{\rm vir} =  \left[ \ov{\mc M}_{\sf \Gamma_1} \right]^{\rm vir} \otimes  \left[ \ov{\mc M}_{\sf \Gamma_2} \right]^{\rm vir}. 
\eeqn
Further, using the K\"unneth formula we decompose 
\beqn
{\rm PD} ( \Delta_X )  = \sum_{i=1}^l \delta_i \otimes \delta^i \in H^* ( X; {\mb Q}  ) \otimes H^*  ( X; {\mb Q} ),
\eeqn
and
\beqn
\iota_{\Pi}^* \beta = \sum_{j=1}^m \beta_j \otimes \beta^j \in H^* ( \ov{\mc M}_{{ \Gamma}_1}; {\mb Q} ) \otimes H^* ( \ov{\mc M}_{{ \Gamma}_2}; {\mb Q} ).
\eeqn
Moreover, let the tails belonging to the $\Gamma_1$ side have indices $i_1  < \cdots < i_{n_1}$ and other tails have indices $j_1 < \cdots < j_{n_2}$. Then one obtains
\beqn
\begin{split}
&\ \left\langle  [ \ov{\mc M}_{\sf \Gamma} ]^{\rm vir}, \alpha_1 \otimes \cdots \otimes \alpha_n \otimes \beta \cup \gamma_\Pi	 \right\rangle \\
&\ \ \ = \sum_{1 \leq i \leq l}^{1 \leq j \leq m} \epsilon \left( \begin{array}{c} \alpha_1, \cdots, \alpha_n \\ i_1, \ldots, i_{n_1}, j_1, \ldots, j_{n_2} \end{array} \right) \left\langle   [ \ov{\mc M}_{\sf \Gamma_1}  ]^{\rm vir}, \alpha_{i_1} \otimes \cdots \otimes \alpha_{i_{n_1}} \otimes \delta_{i} \otimes \beta_j \right\rangle \\
&\ \ \ \ \ \ \ \ \cdot \left\langle   [ \ov{\mc M}_{\sf \Gamma_1}  ]^{\rm vir}, \alpha_{j_1} \otimes \cdots \otimes \alpha_{j_{n_2}} \otimes \delta^i \otimes \beta^j \right\rangle.
\end{split}
\eeqn
Here $\epsilon (\cdot) \in \{\pm 1\}$ is the sign of the permutation of odd degree classes in $\alpha_1, \ldots, \alpha_n$ (the parities of the degrees are considered before the degree shifting). Summing over all equivariant curve classes $\ubar B$, one obtains the splitting property for separating nodes. 

\subsection{Relation with Gromov--Witten theory}\label{subsection62}

An interesting and important question to ask is what the relation between the GLSM correlation functions defined here and the Gromov--Witten invariants.

For convergence concern, take the subring of the Novikov field
\beqn
\Lambda_+  = \Big\{ \sum_{i=1}^\infty a_i T^{\lambda_i} \in \Lambda\ |\ \lambda_i > 0 \Big\}.
\eeqn
For $\alpha \in H^*(X; \Lambda_+)$, we have the {\it GLSM potential}
\beqn
\tilde \tau_g(\alpha) =  \sum_{n \geq 0} \langle \alpha \otimes \cdots \otimes \alpha; 1 \rangle_{g, n}^{\rm GLSM}
\eeqn
and the {\it Gromov--Witten potential}
\beqn
\tau_g(\alpha) = \sum_{n \geq 0} \langle \alpha \otimes \cdots \otimes \alpha; 1 \rangle_{g, n}^{\rm GW}.
\eeqn
The corresponding invariants are just derivatives of $\tilde \tau_g$ or $\tau_g$ at $\alpha = 0$, namely
\begin{align*}
&\ \langle \alpha_1, \cdots, \alpha_n \rangle_{g,n}^{\rm GW} = \frac{\partial^n \tau_g}{\partial \alpha_1 \cdots \partial \alpha_n}(0),\ &\ \langle \alpha_1, \ldots, \alpha_n \rangle_{g, n}^{\rm GLSM} = \frac{\partial^n \tilde \tau_g(\alpha)}{\partial\alpha_1 \cdots \partial \alpha_n }(0).
\end{align*}

Recall that Gromov--Witten invariants can be deformed by inserting extra classes, which corresponds to partial derivatives of $\tau_g$ at a nonzero point.\footnote{However this is not true for GLSM since it is not conformally invariant.} In the case of Lagrangian Floer theory such interior insertions give rise to the so-called {\it bulk deformation} (see \cite{FOOO_book, FOOO_toric_2}). Then the relation between GLSM invariants and Gromov--Witten invariants (of the classical vacuum) can be stated as the following conjecture (see also \cite{Tian_Xu_2017}).

\begin{conjecture}\label{conj62}
There is an element ${\mf c} \in H^*( X; \Lambda_+)$ such that 
\beqn
\frac{\partial^n \tilde \tau_g}{\partial \alpha_1 \cdots \partial \alpha_n} (0) = \frac{\partial^n \tau_g}{\partial \alpha_1 \cdots \partial \alpha_n}({\mf c}). 
\eeqn
\end{conjecture}

Here the class ${\mf c}$ stands for ``correction,'' and it should count pointlike instantons, i.e. solutions to the gauged Witten equation over the complex plane. We give a very brief sketch of the proof here and the details will appear in the forthcoming \cite{Limit2}. The argument is based on an extension of the {\it adiabatic limit} argument which originated from the work of Gaio--Salamon \cite{Gaio_Salamon_2005}. Recall that we have chosen a fibrewise area form $\sigma_{g, n}$ on the universal curve over $\ov{\mc M}_{g, n}$ of cylindrical type such that over long cylinders the perimeters are a fixed constant (say 1). This choice is used to define the gauged Witten equation. Replace $\sigma_{g, n}$ by $\epsilon^{-2} \sigma_{g, n}$ and let $\epsilon$ to to zero. For the rescaled metric, the energy functional of a gauged map $(P, A, u)$ reads
\beqn
E_\epsilon(P, A, u) = \frac{1}{2} \int_{\Sigma_{\mc C}} \Big( \| d_A u\|^2 + \epsilon^2 \| F_A\|^2 + \epsilon^{-2} \| \mu(u) \|^2 + 2 \epsilon^{-2} \| \nabla W(u) \|^2 \Big) {\rm vol}_{\Sigma_{\mc C}}.
\eeqn
Here the integral and the norms are still taken with respect to the non-rescaled metrics, and we also rescale $W$ to $\epsilon^{-1} W$. Then the gauged Witten equation becomes 
\begin{align*}
&\ \ov\partial_A u + \epsilon^{-1} \nabla W(u) = 0,\ &\ * F_A + \epsilon^{-2} \mu(u) = 0.
\end{align*}

Then one can prove a theorem about compactness in the adiabatic limit, saying that for a given sequence of $\epsilon_k \to 0$ and a sequence of solutions $(P_k, A_k, u_k)$ with a uniform energy bound, a subsequence converges in a suitable sense to a holomorphic curve to $X$ away from finitely many points in the domain. Moreover, near each one of the finitely many points where the convergence does not hold, by rescaling the sequence converges to a {\it pointlike instanton}, namely a solution to the gauged Witten equation
\begin{align*}
&\ \ov\partial_A u + \nabla W(u) = 0,\ &\ * F_A + \mu(u) = 0
\end{align*}
over the complex plane ${\mb C}$. These instantons are generalizations of {\it affine vortices} studied in the context of vortex equation in \cite{Gaio_Salamon_2005} \cite{Ziltener_Decay, Ziltener_thesis, Ziltener_book} \cite{VW_affine} \cite{Venugopalan_Xu}. Counting pointlike instantons gives the class ${\mf c}$ in the same way as the correction term constructed by the second-named author with Woodward in \cite{Woodward_Xu} for Lagrangian Floer theory of GIT quotients. 

To prove Conjecture \ref{conj62} one also needs to prove an inverse of the compactness theorem. The ${\mf c}$-deformed Gromov--Witten invariants of $X$ can be thought of as counting holomorphic curves in $X$ with ambient bubble trees, which is a singular object in the $\epsilon = 0$ boundary of the mixed moduli space. One needs to show that in the $\epsilon = 0$ boundary the mixed moduli space has a collaring neighborhood parametrized by $\epsilon$, in the virtual sense. Such a gluing construction can be generalized from that in \cite{Xu_glue}, although the analysis could be much more involved.

\subsection{Relation with the mirror symmetry}

Mirror symmetry predicts that the A-twisted topological string theory of a Calabi--Yau manifold $M$ is isomorphic to the B-twisted topological string theory of its mirror Calabi--Yau manifold $M^\vee$, in particular there is an identification between the A-side moduli (K\"ahler parameters) and the B-side moduli (complex parameters). The identification is generally called the {\it mirror map}. In the case that $M$ is realized as the classical vacuum of a gauged linear sigma model $(V, G, W, \mu)$, many evidences indicate that the mirror map can be understood purely from the A-side, and it counts the correction between the GLSM and the low-energy nonlinear sigma model (the Gromov--Witten theory) (see \cite{Morrison_Plesser_1995}). A reasonable explanation of this phenomenon was given by Hori--Vafa \cite{Hori_Vafa}, which says that the mirror symmetry for vector spaces is trivial and the only nontrivial effect for the compact Calabi--Yau comes from the pointlike instanton correction. From such a principle, there have been many A-side calculations of the mirror map, for example, \cite{Givental_96} \cite{CLLT} \cite{Gonzalez_Woodward_mmp} \cite{CK_2020}. In particular, the counting of affine vortices as done in \cite{Gonzalez_Woodward_mmp}, provides the A-side mirror map interpretation for gauged linear sigma model spaces with zero superpotential. Our forthcoming work \cite{Limit2} will then give the foundation of the general case with nonzero superpotential, following the idea of \cite{Morrison_Plesser_1995}.

\bibliography{mathref}

\bibliographystyle{amsplain}

\end{document}